\documentclass[12pt,draftclsnofoot,onecolumn]{IEEEtran}
\usepackage[OT1]{fontenc}
\usepackage{cite}
\usepackage{dsfont}
\usepackage{amsfonts}
\usepackage{amsmath}
\allowdisplaybreaks[3]
\usepackage{amssymb}
\usepackage{amsthm}
\newtheorem{thm}{Theorem}
\newtheorem{col}{Corollary}

\usepackage{enumitem}
\usepackage{graphicx}
\usepackage{subfigure}
\usepackage{xcolor}
\usepackage{caption}
\usepackage[bookmarks=false, draft]{hyperref}
\usepackage{stfloats}
\usepackage{float}
\usepackage{multirow}
\usepackage[section]{placeins}

\begin{document}
\title{Resource Allocation for Channel Estimation in Reconfigurable Intelligent Surface-Aided Multi-Cell Networks}
	\author{Yining Xu,
          Sheng~Zhou
     \thanks{TThis work was supported in part by the National Natural Science Foundation of China under Grants 62341108, 62022049 and 62111530197.}
     \thanks{Y. Xu and S. Zhou are with Department of Electronic Engineering, Tsinghua University, Beijing 100084, China, and the the Beijing National Research Center for Information Science and Technology (e-mail: xu-yn16@mails.tsinghua.edu.cn; sheng.zhou@tsinghua.edu.cn).}
         }
\maketitle

\begin{abstract}
Reconfigurable intelligent surface (RIS) is a promising solution to deal with the blockage-sensitivity of millimeter wave band and reduce the high energy consumption caused by network densification. However, deploying large scale RISs may not bring expected performance gain due to significant channel estimation overhead and non-negligible reflected interference. In this paper, we derive the analytical expressions of the coverage probability, area spectrum efficiency (ASE) and energy efficiency (EE) of a downlink RIS-aided multi-cell network. In order to optimize the network performance, we investigate the conditions for the optimal number of training symbols of each antenna-to-antenna and antenna-to-element path (referred to as \emph{the optimal unit training overhead}) in channel estimation. Our study shows that: 1) RIS deployment is not `the more, the better', only when blockage objects are dense should one deploy more RISs; 2) the coverage probability is maximized when the unit training overhead is designed as large as possible; 3) however, the ASE-and-EE-optimal unit training overhead exists. It is a monotonically increasing function of the frame length and a monotonically decreasing function of the average signal-to-noise-ratio (in the high signal-to-noise-ratio region). Additionally, the optimal unit training overhead is smaller when communication ends deploy particularly few or many antennas.

\end{abstract}

\IEEEpeerreviewmaketitle
\section{Introduction}
To satisfy the ever growing needs for data rate, seamless coverage and energy efficiency (EE) of wireless communication networks, new technologies emerge, including millimeter wave (mmWave), ultra-dense networks (UDNs) and massive multiple-input multiple-output (MIMO). Yet the coverage demand can not always be met since mmWave is sensitive to blockages due to its short wavelength \cite{bai2014analysis}. In addition, the high energy consumption owing to network densification is still one of the vital problems to be solved. As a promising way to circumvent the high susceptibility of mmWave to blockages and lower down the network energy consumption, reconfigurable intelligent surface (RIS) is slated to play an important role in the next-generation mobile systems \cite{Di2020reconfigurable}. An RIS is constructed with man-made surfaces of electromagnetic materials, e.g., conventional reflect-arrays, liquid crystal surfaces and software-defined meta-surfaces, and thus the flexible control of electromagnetic propagation environment can be realized \cite{Basar2019wireless}. Moreover, RISs can offer significant economic and energy advantages of low-cost. At the same time, unlike relays, RISs do not require complex processing and encoding/decoding, leading to `zero-delay' reflection \cite{wu2020towards}. 

Despite these appealing potentials, RISs can reflect and so enlarge the interference, especially under the dense deployment of RISs \cite{xu2021impact}. Although beamforming reduces the interference in the UDN, how directional transmissions behave in the RIS-aided multi-cell network and how many RISs should be deployed are to be revealed. On top of that, the acquisition of channel state information (CSI) in the RIS-aided multi-cell network becomes a serious expense, in particular when RISs are equipped with a vast number of passive elements \cite{zheng2021survey}. Allocating appropriate resources to balance the CSI accuracy and the overhead so as to promote the network performance, is of great importance.
 
\subsection{Related works}
For RIS channel estimation, several schemes have been proposed and the corresponding overhead has been estimated. An elements grouping method to reduce the training overhead is proposed in \cite{yang2020intelligent}\cite{zhang2020intelligent}, where only the combined channels of each group of elements are estimated and the overhead in terms of pilot transmission time is shown to be proportional to the number of groups. Transmit power allocation and RIS reflection coefficients are optimized to maximize the achievable rate in \cite{yang2020intelligent}. Furthermore, RIS reflection patterns to aid channel estimation are designed and a closed-form expression of the channel estimation error is derived in \cite{zhang2020intelligent}. In \cite{wang2021joint}, the authors propose a joint beam training and positioning scheme, in which random beamforming and maximum likelihood estimation are performed to acquire angle-of-arrival (AoA) and angle-of-departure (AoD). Then an iterative positioning algorithm is applied and the location information can further cross verify the estimation of AoA and AoD. The pairwise error probability of AoA/AoD is proved to be inversely proportional to the training overhead, i.e., the number of channel measurements, in \cite{wang2021joint}. Ref. \cite{zappone2021overhead} develops an overhead-aware resource allocation framework and optimizes the rate and EE w.r.t. RIS phase shifters, transmit and receive filters, as well as power and bandwidth allocation. The time and power overhead of channel estimation are proportional to the product of the number of antennas/elements in base stations (BSs), user equipments (UEs) and RISs. While the length of feedback phase depends on the number of feedback bits for each RIS element and the number of RIS elements \cite{zappone2021overhead}. Furthermore, based on \cite{zappone2021overhead}, the number of RIS elements is optimized when BSs and UEs are equipped with a single antenna \cite{zappone2021optimal}. The results show that rate increases with the number of RIS elements increases, and EE is a concave function of the number of RIS elements\cite{zappone2021optimal}. In \cite{wang2020channel}, the strong correlation of UE-RIS-BS uplink reflected channels between different UEs is utilized to reduce the channel estimation time. Total $KMN+KM$ channel coefficients can be perfectly recovered by $K+N+\max(K-1,\lceil (K-1)N/M)\rceil$ pilot symbols, where $K$, $N$ and $M$ are the number of UEs, RIS elements and BS antennas, respectively \cite{wang2020channel}. These works have effectively reduced the channel estimation overhead, nonetheless, the limited training/feedback resources and the noise/interference still result in imperfect CSI in practice \cite{zheng2021survey}. Consequently, optimizing wireless resources for training/feedback/transmission to balance the accuracy and the overhead of CSI is still of great importance. Most works focus on the performance optimization for a single cell or several adjacent cells, and thus a system-level analysis with large-scale deployed BSs and RISs is needed.

Among the works on system-level performance analysis of RIS-aided networks using stochastic geometry, ref. \cite{kishik2021exploiting} derives the expressions of the average ratio of blind-spot area and the probability distribution of path loss. The results indicate that deploying RISs notably improves the coverage performance when the interference is neglected \cite{kishik2021exploiting}. In \cite{nemati2020ris}, the signal-to-interference-ratio (SIR) coverage probability and the peak reflection power of RISs are derived in closed-forms. The results show that deploying RISs are as effective as equipping BSs with more antennas. Whereas, considering the reflection of interference by RISs, there exists an optimal density of RISs, suggesting that RIS deployment is not `the more, the better', especially when BSs and UEs are equipped with multi-antennas, as revealed in our previous work \cite{xu2021impact}. Ref. \cite{lyu2021hybrid} mainly characterizes the achievable throughput in a hybrid wireless network comprising both active BSs and passive RISs. It is demonstrated that deploying distributed RISs significantly boosts signal power but with marginal interference increases under single antenna BSs and UEs. Apart from that, with a total deployment cost constraint, the optimal RIS-to-BS density ratio that maximizes the network throughput is observed. The authors in \cite{hou2022mimo} study a single-cell scenario with multi-users, where the high signal-to-noise-ratio (SNR) slopes of ergodic rate and the diversity order of outage probability are derived, showing that increasing the number of RIS elements improves the spectrum efficiency and EE. Ref. \cite{psomas2021association} targets at comparing the performance of different UE-RIS association policies, e.g., random association, the closest association and all available association. Results show that the performance comparison heavily relies on the number of RIS elements, cell radius and blockage density. An RIS-aided multi-cell non-orthogonal multiple access (NOMA) network is investigated in \cite{zhang2022reconfigurable} and \cite{zhang2021multi}. The distribution of angle-of-incident and angle-of-reflection are studied, and the coverage probability of paired NOMA users is derived in \cite{zhang2022reconfigurable}. It is concluded that enlarging the size of RIS increases the achievable rate till an upper limit. Besides, ref. \cite{zhang2021multi} considers inter-cell interference and derives the closed-form coverage probability for paired NOMA users. The results evidence that strengthened channel quality overtakes the interference introduced by RISs. Although the performance of RIS-aided networks has been investigated in these previous works, the studies on reflected interference and its impacts on large-scale networks are still in its infancy, especially in terms of the coverage probability, area spectrum efficiency (ASE) and EE.

\subsection{Contributions}
In this paper, we study the large-scale network performance in the light of the coverage probability, the ASE and the EE in an RIS-aided multi-cell network with downlink directional transmissions. In particular, via stochastic geometry, the impact of reflected interference from RISs is analyzed. The optimal number of training symbols per path, referred to as \emph{the optimal unit training overhead}, to maximize the coverage probability, the ASE and the EE is also studied. Our main contributions are summarized as follows:
\begin{itemize}
	\item We utilize a general model of resource allocation for channel estimation and data transmission. The beam alignment error, due to limited channel estimation resources and thereby imperfect CSI, is also characterized. Particularly, the impact of additional interference, from the reflection of RISs, is modeled and studied.
    \item We derive the analytical expressions of the coverage probability, the ASE and the EE for a typical user in the network. Moreover, the monotonicity of the coverage probability w.r.t. unit training overhead in channel estimation is proved. And the condition for the ASE-optimal (also proved to be the EE-optimal) unit training overhead is derived. As a result, the performance of the network can be optimized with the proposed resource allocation scheme.
    \item We adopt extensive numerical studies and find that: 1) The optimal RIS deployment fraction, denoted as the ratio of RIS density to blockage density, is a monotonically increasing function of blockage density and a monotonically decreasing function of BS density. These results indicate that RIS deployment is not `the more, the better', only when blockage objects are dense should one deploy more RISs. 2) The optimal unit training overhead to maximize the ASE and the EE is a monotonically increasing function of the frame length and a monotonically decreasing function of the average SNR (in the high SNR region). Then, the optimal unit training overhead is significantly smaller under especially few or many antennas deployment than that under medium scale of antennas deployment. 
\end{itemize}

\subsection{Organization}
The organization of the rest is as follows. System model is introduced in Section \ref{secmodel}. The expressions of the coverage probability, the ASE and the EE, together with the condition for the optimal unit training overhead, are derived in Section \ref{secresult}. The numerical results in Section \ref{secnum} verify the theoretical results. Finally, Section \ref{seccon} concludes our work.

\section{System Model}\label{secmodel}
\subsection{Network Deployment}
\begin{figure}[t]
  \centering
  \includegraphics[scale = 0.85]{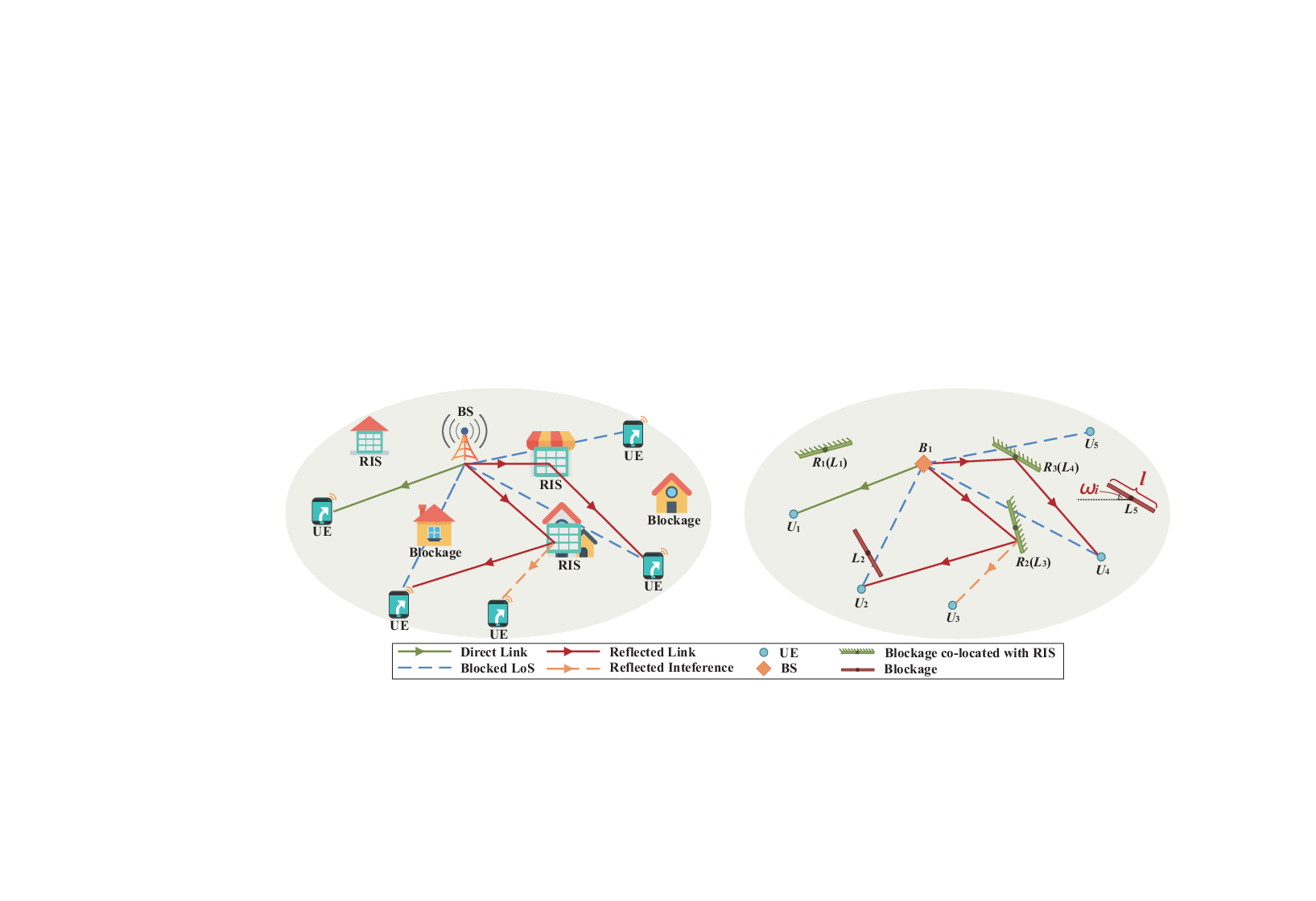}
  \setlength{\abovecaptionskip}{0pt}
  \setlength{\belowcaptionskip}{-20pt}
  \caption{An illustration of network deployment, communication links and stochastic geometry models.}
  \label{Deployment}
\end{figure}
Consider the downlink of a multi-cell network with directional transmissions. We assume that the locations of BSs and UEs follow two independent homogeneous Poisson point processes (HPPPs) $\phi_\text{B} = \{B_i\} \in \mathbb{R}_2$ and $\phi_\text{U} = \{U_i\} \in \mathbb{R}_2$ with density $\lambda_\text{B}$ and $\lambda_\text{U}$, respectively. The attributes of blockages, such as location, orientation and size, are determined by a simplified Boolean model \cite{bai2015coverage}. The blockages are assumed to be line segments with equal length $l$. The center points of the blockages follow another HPPP $\phi_\text{L}= \{L_i\} \in \mathbb{R}_2$ with density $\lambda_\text{L}$, which is independent of $\phi_\text{B}$ and $\phi_\text{U}$. The orientation of the blockage $L_i$ is an uniformly distributed random variable $\omega_i$ in the range of $[0,2\pi)$, which is independent and identically distributed (i.i.d.). We consider an RISs deployment strategy that a fraction of blockages are co-located with (a.k.a. replaced by) RISs, as shown in Fig.\ref{Deployment}. Therefore, the center points of RISs, denoted by $\phi_\text{R}= \{R_i\} \in \mathbb{R}_2$, actually follow a thinning process of $\phi_\text{L}$ with density $\lambda_\text{R} = \mu \lambda_\text{L}$, where $\mu \in [0,1]$ is the RIS deployment fraction. We assume that RISs and blockages are of the same size, and only one of the two surfaces of an RIS is the reflection surface.

\subsection{Channel Model: Blockage, Path Loss and Small-Scale Fading}
For any communication link with distance $r$, the line-of-sight (LoS) probability is $p_\text{L}(r) = \text{exp}(-\eta r)$, where $\eta = \frac{2l\lambda_\text{L}}{\pi}$  under the aforementioned Boolean blockage model \cite{bai2015coverage}. We denote the LoS indicator  as $\mathbb{I}_\text{L}(r)$, where $\mathbb{I}_\text{L}(r) = 1$ represents LoS communication link and $\mathbb{I}_\text{L}(r) = 0$ indicates the occurrence of blockage. Therefore, the distribution of $\mathbb{I}_\text{L}(r)$ is
\begin{align}\label{IL}
  \mathbb{I}_\text{L}(r)=
  \begin{cases}
    1 & \text{w.p.  } p_\text{L}(r)\\
    0 & \text{w.p.  } 1-p_\text{L}(r).
  \end{cases}
\end{align}

We consider two types of communication links, as shown in Fig.\ref{Deployment}. When the BS-UE link is LoS, i.e., the UE associates to the BS directly, the path loss of a direct link with distance $r_{\text{BU}}$ is $P\!L_\text{D}(r_{\text{BU}}) = r_{\text{BU}}^{-\alpha}$, where $\alpha$ is the path loss exponent. When the direct link is blocked, the RIS can provide an reflected link to the blocked UE. The sum-distance path loss model is used since mmWave wavelength is sufficiently small as compared with RIS element size \cite{di2020analytical}. In addition, the energy loss after the reflection of an RIS is modeled by a reflection power attenuation coefficient $\gamma$. In that case, the path loss of a reflected link, with distance $r_{\text{BR}}$ between the BS and the RIS and distance $r_{\text{RU}}$ between the RIS and the UE, is $P\!L_\text{R}(r_{\text{BR}}+ r_{\text{RU}}) = \gamma(r_{\text{BR}}+r_{\text{RU}})^{-\alpha}$, where $r_{\text{BR}}+ r_{\text{RU}}$ is the equivalent reflected link distance. 

We assume a Rayleigh small-scale fading and the channel fading coefficient is denoted as $h \sim \text{exp}(1)$. The channel fading coefficient $h$ is i.i.d. for direct links and reflected links.

\subsection{Feasibility of Reflection}
\begin{figure}[tb]
  \centering
  \includegraphics[scale = 0.47]{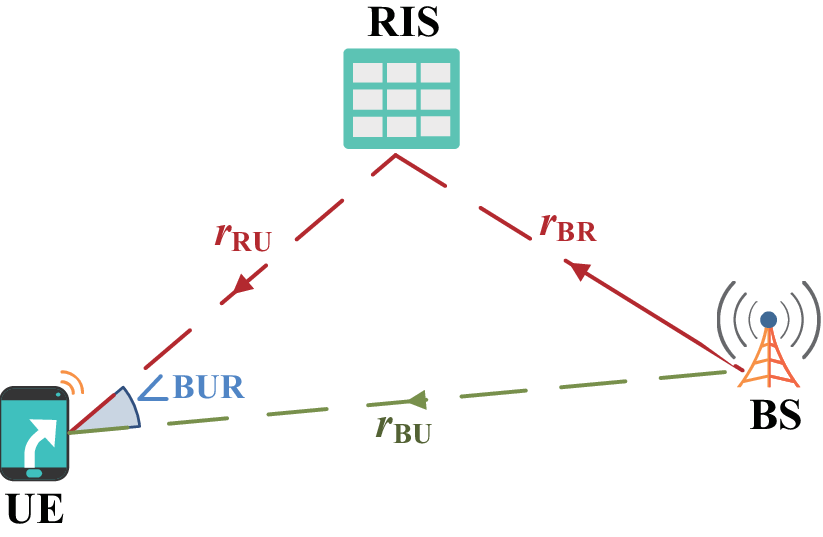}
  \setlength{\abovecaptionskip}{0pt}
  \setlength{\belowcaptionskip}{-20pt}
  \caption{An illustration of feasible reflection.}
  \label{feas}
\end{figure}
 As illustrated in Fig. \ref{feas}, the feasibility of the reflected link depends on whether the BS and the UE are on the same side of the RIS reflection surface. We denote the probability of feasible reflection as $p_\text{F}(r_{\text{BU}},r_{\text{RU}},\angle{\text{BUR}})$, where $\angle{\text{BUR}}$ is the angle between the UE-BS path and the UE-RIS path. The expression of $p_\text{F}(r_{\text{BU}},r_{\text{RU}},\angle{\text{BUR}})$ is derived in \cite{kishik2021exploiting} as
\begin{align}
   \begin{split}
  p_\text{F}(r_{\text{BU}},r_{\text{RU}},\angle{\text{BUR}}) 
  = \frac{1}{2}
     \left(\!1\!-\! \frac{1}{\pi} 
     \cos^{-1}\!\!\left(\frac{r_{\text{RU}}-r_{\text{BU}} \cos(\angle{\text{BUR}})}{\!\!\sqrt{\!r_{\text{BU}}^2 \!+ \!r_{\text{RU}}^2 \!-\! 2r_{\text{BU}}r_{\text{RU}} \cos(\angle{\text{BUR}})}}\right)\!\!\right).\\
   \end{split}
  \end{align}
  Similarly, we also denote a feasible reflection indicator $\mathbb{I}_\text{F}(r_{\text{BU}},r_{\text{RU}},\angle{\text{BUR}})$ to indicate whether the reflected link is feasible or not. The distribution of $\mathbb{I}_\text{F}(r_{\text{BU}},r_{\text{RU}},\angle{\text{BUR}})$ is given by
  \begin{align}\label{IF}
  \mathbb{I}_\text{F}(r_{\text{BU}},r_{\text{RU}},\angle{\text{BUR}})=
  \begin{cases}
    1 & \text{\!\!w.p. } p_\text{F}(r_{\text{BU}},r_{\text{RU}},\angle{\text{BUR}})\\
    0 & \text{\!\!w.p. } 1\!-\!p_\text{F}(r_{\text{BU}},r_{\text{RU}},\angle{\text{BUR}}).
  \end{cases}
\end{align}

\subsection{Beamforming Pattern}
\begin{figure}[tbh]
  \centering
  \hspace{-0.1in}
  \includegraphics[scale = 0.53]{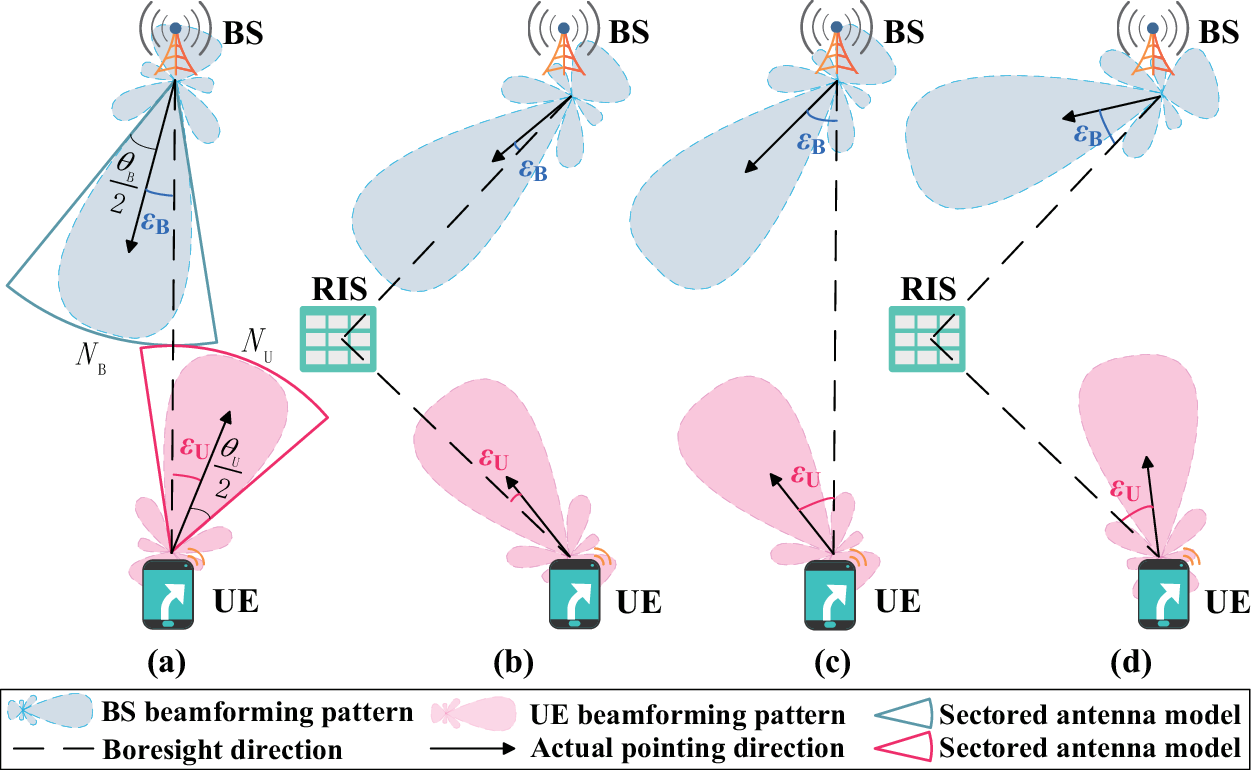}
  \setlength{\abovecaptionskip}{3pt}
  \setlength{\belowcaptionskip}{-20pt}
  \caption{An illustration of beamforming patterns and beam alignment. In (a) and (b), the beam alignment error is no larger than half beamwidth, therefore the BS and the UE are aligned through a direct link and a reflected link, respectively. In (c) and (d), the beam alignment error exceeds half beamwidth, then the BS and the UE are misaligned in both the direct link and the reflected link, respectively.}
  \label{Beam}
\end{figure}

BSs and UEs are equipped with uniform linear arrays (ULAs) with $M_\text{B}$ and $M_\text{U}$ antennas, respectively. The beamforming pattern is approximated to a sectored antenna model for tractability, as shown in Fig. \ref{Beam}. The side lobe gain can be neglected under high front-back beamforming gain ratio. As a result, the beamforming pattern is defined by two parameters: main lobe gain $N_j$ and beamwidth $\theta_j$, where $j\in\{\text{B},\text{U}\}$ represents BSs or UEs, respectively. The total radiation gain constraint $N_j \theta_j = 2 \pi$ is satisfied under varying beamforming patterns. For the ULAs with half-wavelength antenna separation, we have $\theta_j = \frac{4}{M_j}$ in radians.

We also assume that the number of reflection elements in an RIS is $M_\text{R}$. Since RISs are not equipped with power amplifier, there is no power gain after the reflection of RISs. And the constant transmit (or reflect) power consumption of a BS and an RIS are denoted as $P_\text{B}$ and $P_\text{R}$, respectively.

\subsection{User Association}\label{userasso}
Without loss of generality, we focus on a typical user located at the origin $o$. The typical user associates to the BS with the strongest average received signal power. The UEs, which associate to the same BS, are allocated with orthogonal time-frequency resource blocks (RBs). Therefore, intra-cell interference is zero. When the typical user associates to a BS $B_i$ through a direct link, the average received signal power is
\begin{align}
  S_{B_i}(r_{B_i}) = P_\text{B}\mathbb{I}_\text{L}(r_{B_i}) P\!L_\text{D}(r_{B_i}),
\end{align}
where $r_{B_i}$ is the distance between the BS $B_i$ and the typical user. While for an association BS $B_i$ with a reflected link through an RIS $R_k$, the average received signal power of the typical user is
\begin{align}
 \begin{split}
  S_{B_i\!R_k}(r_{B_i}, r_{R_k}, r_{B_i\!R_k}, \angle{B_ioR_k})
  = P_\text{B}
   \mathbb{I}_\text{L}(r_{B_i\!R_k})  
   \mathbb{I}_\text{L}(r_{R_k})  
   \mathbb{I}_\text{F}(r_{B_i},r_{R_k},\angle{B_ioR_k})   
   P\!L_\text{R}(r_{B_i\!R_k}+r_{R_k}), 
  \end{split}
\end{align}
where $r_{R_k}$ is the distance between the RIS $R_k$ and the typical user, and $\angle{B_ioR_k}$ is the angle between the direction from the origin $o$ to the BS $B_i$ and the direction from the origin $o$ to the RIS $R_k$, and $r_{B_i\!R_k} = \sqrt{r_{B_i}^2+r_{R_k}^2-2r_{B_i}r_{R_k} \cos (\angle{B_ioR_k})}$ is the distance between the BS $B_i$ and the RIS $R_k$, as illustrated in Fig. \ref{feas} with $B \rightarrow B_i$, $R \rightarrow R_k$, $U \rightarrow o$, $r_\text{BU} \rightarrow r_{B_i}$ and $r_\text{RU} \rightarrow r_{R_k}$, specifically. 

Therefore, the association BS $B^*$ is determined by
\begin{align}
 \begin{split}
  B^* = \mathop{\arg\max}_{B_i \in \phi_{\text{B}}}
        \max
        \left(\!S_{B_i}(r_{B_i}), 
            \mathop{\max}_{R_k\in \phi_{\text{R}}}
            \!\!S_{B_i\!R_k}(r_{B_i}, r_{R_k}, r_{B_i\!R_k}, \angle{B_ioR_k})
        \!\!\right).
 \end{split}  
\end{align}
This association policy attempts to find an LoS and feasible communication link with the lowest path loss. In spite of that, UEs may experience association failure when all the communication links are blocked by obstructions or infeasible for reflection. In that case, the average received signal power of the typical user is zero and a covering hole appears.

\subsection{Frame Structure and Resource Allocation}\label{framestruc}
\begin{figure}[tb]
  \centering
  \includegraphics[scale = 0.45]{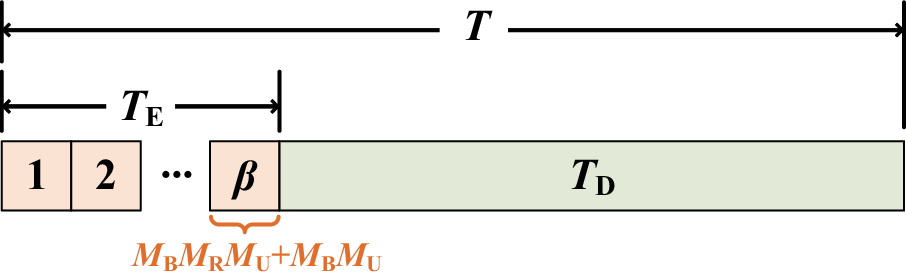}
  \setlength{\abovecaptionskip}{0pt}
  \setlength{\belowcaptionskip}{-20pt}
  \caption{Frame structure with channel estimation phase and data transmission phase.}
  \label{frame}
\end{figure}
We divide a transmission frame of fixed length $T$ into two sub-phases: the channel estimation phase of length $T_\text{E}$ and the data transmission phase of length $T_\text{D}$, as shown in Fig. \ref{frame}.

In the channel estimation phase, a BS transmits $T_\text{E}$ shared pilots to a UE and CSI is acquired at the UE through downlink training. After that, CSI is provided to the BS by feedback, and the overhead of feedback is regarded as a constant. Here we neglect the feedback overhead to simplify our analyses. Both the direct link and the reflected link of the association BS are estimated. Overall, the channel estimation overhead is modeled as
\begin{align}
  T_\text{E} = \beta(M_\text{B}M_\text{R}M_\text{U}+M_\text{B}M_\text{U}),
\end{align}
where $\beta$ is the overhead of estimating a path from a single antenna (element) to another in the corresponding communication end. The parameter $\beta$ is the number of training symbols per path, and we refer to it as the unit training overhead. Note that, $\beta$ is in the range of $[0,\frac{T}{M_\text{B}M_\text{R}M_\text{U}+M_\text{B}M_\text{U}})$. A non-integer $\beta$ can be regarded as a result of applying channel estimation overhead reduction schemes, e.g., the RIS elements grouping scheme \cite{yang2020intelligent}\cite{zhang2020intelligent}. Despite of that, the proposed channel estimation overhead model can work as a worst case analysis \cite{wang2020channel}. Furthermore, $\beta = 0$ means that we do not transmit any pilot and choose the beam direction randomly. As a result, the length of the data transmission phase is $T_\text{D} = T-T_\text{E}$.

On the other hand, the accuracy of CSI influences the precision of beam alignment \cite{park2012outage}. Under minimum mean square error (MMSE) estimation, the channel estimation error is Gaussian distributed with variance $\sigma_\text{E}^2 = \frac{1}{1+\beta \text{SNR}}$, where $\text{SNR}$ denotes the average SNR of channels \cite{caire2010multiuser}. We denote the beam alignment error $\varepsilon_j$ with $j \in \{\text{B},\text{U}\}$ as the angular difference between the boresight direction and the actual pointing direction of the BS and the UE, respectively, as shown in Fig. \ref{Beam}. Assume that the beam alignment error is a truncated Gaussian distributed random variable with zero mean \cite{cheng2018coverage}\cite{thornburg2015ergodic}, i.e., $\varepsilon_j \sim \mathcal{N}_\text{T}(0, \sigma_j^2,-\pi,\pi)$, where $\sigma_j^2$ is the variance and the variable is within the range of $[-\pi,\pi]$. A natural idea is to tie the channel estimation error and the beam alignment error together. Consequently, we assume that the variance of the beam alignment error $\sigma_j^2$ and the variance of the channel estimation error $\sigma_\text{E}^2$ satisfy the function that $\sigma_j^2 = k_j \pi^2 \sigma_\text{E}^2$ with $k_j \in (0,1]$, where $k_j \pi^2$ is the variance of the beam alignment error without channel estimation. We can revisit this function from two aspects: 1) when $\beta \rightarrow 0$, meaning that we almost choose the beam direction randomly without channel estimation, therefore the variance of the beam alignment error $\sigma_j^2$ approaches $k_j^2 \pi^2$; 2) when $\beta \rightarrow \infty$, meaning that we use enough pilots to estimate the channel, therefore the variance of the beam alignment error $\sigma_j^2$ approaches $0$ and perfect beam alignment is almost reached.

In summary, a larger $\beta$ results in more accurate CSI and a smaller beam alignment error, but the time left for data transmission is reduced simultaneously. How to allocate the resources (either over time or bandwidth) for channel estimation and data transmission is our focus in this paper.

\subsection{Performance Metrics}\label{performmetric}
First, we will discuss the distribution of the total beamforming gain. We denote the probability that the beam alignment error is no larger than half beamwidth as $p_{\text{E}_j}(\sigma_j, \theta_j)$ with $j \in\{\text{B},\text{U}\}$. The expression of $p_{\text{E}_j}(\sigma_j, \theta_j)$ is obtained according to the truncated Gaussian distribution of $\varepsilon_j$, which is 
\begin{align}\label{PE}
p_{\text{E}_j}(\sigma_j, \theta_j) = \mathbb{P}\left[|\varepsilon_j| \leq \frac{\theta_j}{2}\right]
                                   = \frac{\text{erf}\left(\frac{\theta_j}{2 \sqrt{2 \sigma_j^2}}\right)}
                                   {\text{erf}\left(\frac{\pi} {\sqrt{2 \sigma_j^2}}\right)}.
\end{align}
When the beam alignment error is no larger than half beamwidth, the main lobe gain can still be obtained; otherwise, the corresponding communication end only receives zero gain, i.e., the side lobe gain, named misalignment. Fig. \ref{Beam} illustrates the situations of beam alignment and misalignment. Then, the distribution of the total beamforming gain of the link between the typical user and its association BS is expressed as
\begin{align}\label{G0}
G_0 = 
\begin{cases}
    N_\text{B} N_\text{U} & \text{w.p.  } p_{\text{E}_\text{B}}(\sigma_\text{B}, \theta_\text{B})
                                          p_{\text{E}_\text{U}}(\sigma_\text{U}, \theta_\text{U})\\
    0 & \text{w.p.  } 1-p_{\text{E}_\text{B}}(\sigma_\text{B}, \theta_\text{B}) 
                      p_{\text{E}_\text{U}}(\sigma_\text{U}, \theta_\text{U}).
  \end{cases}
\end{align}

However, it is not the case for links between the typical user and interfering BSs. The statistics of interference is invariant with beam misalignment \cite{thornburg2015ergodic}. Consequently, the distribution of the total beamforming gain of interfering links is actually the distribution without beam alignment error, which is
\begin{align}\label{Gi}
G_i = 
\begin{cases}
    N_\text{B} N_\text{U} & \text{w.p.  } \frac{\theta_\text{B}\theta_\text{U}}{4\pi^2}\\
    0 & \text{w.p.  } 1-\frac{\theta_\text{B}\theta_\text{U}}{4\pi^2},
  \end{cases}
\end{align}
and $G_i$ is i.i.d. for each interfering BS $B_i$.

Finally, the expression of received signal-to-interference-and-noise-ratio (SINR) of the typical user is
  \begin{align}
    \begin{split}
    \text{SINR} = \frac{G_0 h 
                    \left[
                    \mathbb{I}_\text{L}(r_{\!B^*}\!) P\!L_\text{D}(r_{B^*}\!)
                     +\! 
                    (1\!-\!\mathbb{I}_\text{L}(r_{\!B^*}\!)\!) 
                    P\!L_\text{R}(r_{B^*\!R^*}\!+\!r_{R^*}) 
                    \right]}
                     {N_0 + I_\text{D} +I_\text{R}},
    \end{split}
  \end{align}
where $N_0$ is the noise power normalized by BS transmits power. And $R^*$ denotes the association RIS when the association link is a reflected link, obtaining from
\begin{align}
    \begin{split}
    R^* = \mathop{\arg \max}_{R_k\in \phi_{\text{R}}}
            S_{B^*\!\!R_k}(r_{B^*}, r_{R_k}, r_{B^*\!\!R_k}, \angle{B^*\!oR_k}).
    \end{split}
  \end{align}
While $I_\text{D}$ and $I_\text{R}$ are normalized direct interference power and normalized reflected interference power, respectively. Their expressions are listed as follows
\begin{align}
&I_\text{D} = \mathop{\sum}_{i: B_i \in \phi_\text{B} \backslash\{B^*\!\}} 
              \!\!\!S_{B_i}(r_{B_i}) G_i h /P_\text{B},\\
&I_\text{R} = \mathop{\sum}_{i: B_i \in \phi_\text{B} \!\backslash\!\{B^*\!\}} 
              \mathop{\sum}_{k: R_k \in \phi_\text{R}} 
              \!\!\!S_{B_i\!R_k}(r_{B_i}, r_{R_k}, r_{B_i\!R_k}, \angle{B_ioR_k}) G_i h /P_\text{B}.
\end{align}

There are three main performance metrics:
\begin{enumerate}
\item  Coverage probability, defined as $\mathcal{P}= \mathbb{P}[\text{SINR}>\tau]$, is the probability that the received SINR of the typical user is larger than the threshold $\tau$. This metric indicates the ability of providing the basic coverage service. 
\item  ASE, defined as $\mathcal{A} = \frac{T_\text{D}}{T} \lambda_\text{B} \mathbb{E}[\text{log}_2 (1+\text{SINR}) \mathbb{I}\{\text{SINR}>\tau\}]$ with indication function $\mathbb{I}\{\cdot\}$, is the average throughput per unit frequency and area. This metric measures how well the network utilizes bandwidth resources.
\item  EE, defined as $\mathcal{E} = \mathcal{A}/(\lambda_\text{B} P_\text{B}+ \lambda_\text{R} P_\text{R})$, is the throughput per unit frequency and power. EE represents how well the network utilizes bandwidth and energy resources. The power consumption of UEs are neglected since it is sufficiently small comparing with the power consumption of BSs and RISs.
\end{enumerate}
 
\section{Main Results}\label{secresult}
Before the derivation of the coverage probability, the ASE and the EE, we first present some preliminary knowledge on link distance distributions and association probability. Detailed proof can be found in \cite{kishik2021exploiting}.

\subsection{Link Distance Distribution}
The distance of the shortest LoS direct link is denoted by a random variable $r_\text{D}$ and its cumulative distribution function (CDF) is $F_{\text{D}}(x)$. The conditional CDF of the distance of the shortest LoS reflected link between the typical user and the BS at distance $u$ is denoted by $F_{\text{R}|u}(x)$ with $x \geq u$. The distance of the shortest LoS reflected link is denoted by a random variable $r_\text{R}$ and its CDF is $F_\text{R}(x)$. 

\subsubsection{The Distribution of the Shortest LoS Direct Link Distance}
We have following expressions
\begin{align}
 \begin{split}
F_\text{D}(x) = & \mathbb{P}[r_\text{D} \leq x]\\
              = & 1-\mathbb{P}[\mathcal{N}_{\phi_\text{B,L}}(\mathcal{B}(o,x))=0] \\
              =& 1- \exp\left(-\int_{\mathcal{B}(o,x)} \lambda_\text{B} p_\text{L}(r)  \text{d} r\right)\\
              =& 1-\exp \left( -2\pi \lambda_\text{B} \frac{1-(\eta x +1) \exp(-\eta x)}{\eta^2}\right), 
 \end{split}
\end{align}
where $\mathcal{N}_\phi(\mathcal{C})$ represents the number of points in the region $\mathcal{C}$ following the PPP $\phi$. And $\phi_\text{B,L}$ is a location-dependent thinning process of $\phi_\text{B}$ with density $\lambda_\text{B} p_\text{L}(r)$, representing an inhomogeneous PPP of the locations of LoS BSs. And $\mathcal{B}(o,x)$ is a disc centered at the origin $o$ with radius $x$. Finally, the corresponding probability density function (PDF) is expressed as
\begin{align}\label{PDF_D}
 \begin{split}
f_\text{D}(x) = \frac{\text{d} F_\text{D}(x)}{\text{d} x}
              =  2\pi \lambda_\text{B}x \exp\left(-\eta x- 2\pi \lambda_\text{B} 
                  \frac{1-(\eta x +1) \exp(-\eta x)}{\eta^2}\right).
 \end{split}
\end{align}

\subsubsection{The Conditional Distribution of the Shortest LoS Reflected Link Distance}
Similar to the derivation of $F_\text{D}(x)$, we can write the expression of $F_{\text{R}|u}(x)$ as
\begin{align}
 \begin{split}
F_{\text{R}|u}(x)
= &  \mathbb{P}[r_\text{R} \leq x|u]\\
= &1-\mathbb{P}[\mathcal{N}_{\phi_\text{R,L}}(\mathcal{C}_1)=0]\\
= & 1 \!-\! \exp \!\left(\!\!-\!\!\int_{\mathcal{C}_1} \!\!\lambda_\text{R} p_\text{L}\!(t)p_\text{L}\!\!\left(\!\!\sqrt{u^2\!+\!t^2\!-\!\!2 u t\cos\psi}\right) \!p_\text{F}(u,t,\psi) t \text{d} t \text{d} \psi\!\!\right)\\
\overset{(a)}= & 1 \!-\! \exp \!\Bigg(\!\!-\lambda_\text{R} \!\!\int_{\!-\!\pi}^{\pi} \!\!\int_0^{\frac{x^2-u^2}{2(\!x\!-\!u\!\cos\!\psi\!)}} 
 \!p_\text{L}\!(t)p_\text{L}\!\!\left(\!\!\sqrt{u^2\!+\!t^2\!-\!\!2 u t\cos\psi}\right) \!p_\text{F}(u,t,\psi) t \text{d} t \text{d} \psi\!\!\Bigg), 
  \end{split}
\end{align}
where $\phi_\text{R,L}$ is an inhomogeneous PPP representing the locations of RISs that provide LoS feasible reflected links from the BS at $u$ to the typical user. The density of $\phi_\text{R,L}$ is expressed as $\lambda_\text{R}p_\text{L}(t)p_\text{L}\!\left(\!\!\sqrt{\!u^2\!\!+\!t^2\!-\!2 u t\!\cos\!\psi}\right)p_\text{F}(u,\!t,\!\psi)$, where $t$ and $\psi$ indicate $r_\text{RU}$ and $\angle{\text{BUR}}$ for simplicity. The points in the region $\mathcal{C}_1$ satisfy that the BS-RIS-UE reflected link distance is no larger than $x$, i.e., $\mathcal{C}_1 = \{ t, \psi : t + \sqrt{u^2\!+\!t^2\!-\!\!2 u t\cos\psi} \leq x\}$, resulting in the range of the integral in (a). 

\subsubsection{The Distribution of the Shortest LoS Reflected Link Distance}
We denote that the locations of the BSs, which are blocked in direct links and have feasible LoS reflected links with link distance no larger than $x$, follow an inhomogeneous PPP $\phi_\text{B,R}$. The density of $\phi_\text{B,R}$ is $\lambda_\text{B}(1-p_\text{L}(u))F_{\text{R}|u}(x)$. Therefore, the CDF of the shortest reflected link distance $r_\text{R}$ is
\begin{align}
 \begin{split}
  F_\text{R}(x)
=  \mathbb{P}[r_\text{R} \!\leq x]
=  1-\mathbb{P}[\mathcal{N}_{\phi_\text{B,R}}(\mathbb{R}^2)=0]
=  1 \!-\! \exp \!\left(\!\!-2\pi\lambda_\text{B}\!\!\int_0^{x}  (1-p_\text{L}\!(u)) F_{\text{R}|u}(x) u \text{d} u\!\!\right).
  \end{split}
\end{align}
And the corresponding PDF is
\begin{align}\label{PDF_R}
 \begin{split}
f_\text{R}(x) 
= \frac{\text{d} F_\text{R}(x)}{\text{d} x}
=  2\pi \lambda_\text{B}
    \exp \!\!\left(
    \!\!- 2\pi  \lambda_\text{B}
    \!\!\int_0^x  \!\!\!(1 \!-\! p_\text{L}(u)\!) F_{\text{R}|u}(x) u \text{d} u \!\!\right) 
    \!\! \int_0^x \!(1\!-\!p_\text{L}(u)\!) \frac{\text{d} F_{\text{R}|u}(x) }{\text{d} x} u \text{d} u,
  \end{split}
\end{align}
where

 \begin{align}
  \begin{split}
\frac{\text{d} F_{\text{R}|u}(x) }{\text{d} x} 
    =& \lambda_\text{R} \text{exp}\!\left(\!-\!\lambda_\text{R} \!\!\!\int_{\!-\!\pi}^{\pi}
    \!\int_0^{\frac{x^2-u^2}{2 \left(\!x\!-\!u \!\cos \!\psi\! \right)}}
    \!\!p_\text{L}\!(t)p_\text{L}\!\!\left(\!\!\sqrt{u^2\!\!+\!t^2\!\!-\!\!2 u t\cos\psi}\right) \!p_\text{F}(u,t,\psi) t \text{d} t \text{d} \psi\!\!\right)\!\cdot \\
    & \quad\!\!\!\!\!\int_{\!-\!\pi}^{\pi} \!\!\!
    \!p_\text{L}\!\!\left(\!\!\frac{x^2-u^2}{2 \!\left(\!x\!\!-\!u \!\cos \!\psi\right)}\!\!\right)
    \!\!p_\text{L}\!\!\!\left(\!\!\sqrt{u^2\!\!+\!\!{\left(\!\frac{x^2-u^2}{2 \!\left(\!x\!-\!u \!\cos \!\psi\! \right)}\!\right)}^{\!\!2}\!\!\!-\!\!\frac{u(x^2\!-\!u^2)}{\!x\!-\!u \!\cos \!\psi\! }\!\cos\psi}\right)
    \!\!p_\text{F}\!\!\left(\!\!u,\!\frac{x^2-u^2}{2 \!\left(\!x\!-\!u \!\cos \!\psi\! \right)},\psi\!\!\right) \!\cdot\\
    &\quad \quad \quad \!\!\frac{(x^2\!\!-\!2 u x \cos \psi \!+\! u^2)(\!x^2\!-\!u^2)}{4(x-u\cos\psi)^3} \text{d} \psi.
   \end{split}
 \end{align}

\subsection{Association Probability}
According to our user association policy introduced in Section \ref{userasso}, the typical user may have three kinds of association states:
\begin{enumerate}
\item The typical user associates to an LoS BS through a direct link, denoted as event $\mathcal{D}$; 
\item The typical user associates to a BS through an RIS reflected link, denoted as event $\mathcal{R}$;
\item The typical user has no BS to associate, denoted as event $\mathcal{O}$.
\end{enumerate}
The probability of these three states are expressed as $\mathbb{P}_\mathcal{D}$, $\mathbb{P}_\mathcal{R}$ and $\mathbb{P}_\mathcal{O}$, respectively. Naturally, $\mathbb{P}_\mathcal{D} + \mathbb{P}_\mathcal{R} + \mathbb{P}_\mathcal{O} = 1$ should be satisfied. 

First, the expression of $\mathbb{P}_\mathcal{O}$ is obtained from
\begin{align}
 \begin{split}
 \mathbb{P}_\mathcal{O} 
 \overset{\text{(a)}}=
  &\mathbb{E}\left[ \prod_{B_i \in \phi_\text{B}}
  \!\!\!\left(\!1\!-\!\left(\mathbb{I}_\text{L}(\!r_{B_i}\!) + (1-\mathbb{I}_\text{L}(\!r_{B_i}\!)) 
  \!\left(\!\!1\!-\!\!\!\! \!\prod_{R_k \in \phi_{\text{R}}} 
  \!\!\!\!\left(1\!-\!\mathbb{I}_\text{L}(\!r_{B_i\!R_k}\!)  
  \mathbb{I}_\text{L}(\!r_{R_k}\!)  
  \mathbb{I}_\text{F}(\!r_{B_i},r_{R_k},\angle{B_ioR_k}\!)\!\right)\!\!\right)
  \!\!\!\right)\!\!\!\right)\!\!\right]\\
 \overset{\text{(b)}}=
  &\mathbb{E}_{\phi_{\text{B}}}\!\!
  \left[ \prod_{B_i \in \phi_\text{B}}
  \!\!\!\left(\!1\!-\!\left(\!\!p_\text{L}(\!r_{B_i}\!) \!+\! (1 \!-\! p_\text{L}(\!r_{B_i}\!)\!) 
  \mathbb{E}_{\phi_{\text{R}}}\!\!
  \left[\!1\!-\!\!\!\! \!\prod_{R_k \in \phi_{\text{R}}} 
  \!\!\!\!\left(1\!-\!p_\text{L}(\!r_{B_i\!R_k}\!)  
  p_\text{L}(\!r_{R_k}\!)  
  p_\text{F}(\!r_{B_i},r_{R_k},\angle{B_ioR_k}\!)\!\right)\!\right]
  \!\right)\!\!\!\right)\!\!\right]\\
 \overset{\text{(c)}}= 
  &\text{exp} \!\Bigg(\!\!\!-\! 2 \pi \lambda_\text{B} 
  \!\!\int_0^{\infty} 
  \!\!\Bigg(\!p_\text{L}\!(u) + (1 \!-\! p_\text{L}\!(u)\!) \cdot \Bigg.\Bigg.\\
  &\Bigg. \Bigg.\quad \quad \!\!\Bigg(\!1-\exp\!\Bigg(\!-\!\lambda_\text{R} 
  \!\!\int_{-\pi}^{\pi} \!\int_0^{\infty} 
  \!\!p_\text{L}\!\left(\!\sqrt{t^2 \!+\! u^2 \!-\! 2 u t \cos \psi}\right) 
  \!p_\text{L}\!(t) p_\text{F}(u,t,\psi) 
  t \text{d}t \text{d} \psi \!\Bigg)\!\!\Bigg)\!\!\Bigg)
  \!u \text{d} u\!\Bigg),
 \end{split}
\end{align}
where (a) results from the fact that neither a direct link nor a reflected link is available for the typical user in a covering hole. Equation (b) results from the independent blockages for each link and the distributions of $\mathbb{I}_\text{L}(r)$ and $\mathbb{I}_\text{F}(r_{\text{BU}},r_{\text{RU}},\angle{\text{BUR}})$ proposed in (\ref{IL}) and (\ref{IF}), respectively. Equation (c) follows by the application of probability generating functional (PGFL) of HPPP with variables substitution.

Then, we turn to the expression of $\mathbb{P}_\mathcal{D}$. The typical user associates to a BS through a direct link, that is to say, the association link is the nearest LoS direct link and any other reflected link has higher path loss. Conditioned on that the nearest LoS BS locates at distance $r_\text{D}$, the locations of BSs form an inhomogeneous PPP $\phi_{\text{B,R}|r_\text{D}}$, where the BSs are blocked in direct links and provide reflected links with stronger average received signal power than the nearest LoS BS. Therefore, the density of $\phi_{\text{B,R}|r_\text{D}}$ is $\lambda_\text{B} (1-p_\text{L}(u)) F_{\text{R}|u}(r_\text{D}\gamma^{\frac{1}{a}})$.
According to the definition of $\mathbb{P}_\mathcal{D}$, we have that
\begin{align}\label{XD}
 \begin{split}
\mathbb{P}_\mathcal{D}
=& \mathbb{E}_{r_\text{D}}\!\Big[\mathbb{P}[\mathcal{N}_{\phi_{\text{B,R}|r_\text{D}}}(\mathbb{R}^2)=0|r_\text{D}]\Big]\\
=& \mathbb{E}_{r_\text{D}}\!\!\left[\exp \!\left(\!-2\pi\lambda_\text{B}
   \!\!\int_0^{\infty} \!\!\!(1-p_\text{L}(u)) F_{\text{R}|u}(r_\text{D}\gamma^{\frac{1}{a}})
   u \text{d} u\!\right)\!\right]\\
=& \!\!\int_0^{\infty} \!\!\!\!f_\text{D}(x) 
   \exp \!\left(\!-2\pi\lambda_\text{B}
   \!\!\int_0^{\infty} \!\!\!(1-p_\text{L}(u)) F_{\text{R}|u}(x\gamma^{\frac{1}{a}})
   u \text{d} u \!\right)
   \!\text{d} x.
 \end{split}
\end{align}
Eventually, the probability of reflected links association is obtained
\begin{align}\label{XR}
\mathbb{P}_\mathcal{R} = 1-\mathbb{P}_\mathcal{O}-\mathbb{P}_\mathcal{D}. 
\end{align}

In the next step, we derive the expressions of the coverage probability, the ASE and the EE based on these preliminary knowledge on link distance distributions and association probability.

\subsection{Coverage Probability}
\begin{figure*}[t]
\hrule height 1pt
\vspace*{2pt}
\begin{thm}\label{thm1}
  The coverage probability of the downlink RIS-aided multi-cell network with directional transmissions is
 \setlength{\abovedisplayskip}{2pt}
 \setlength{\belowdisplayskip}{2pt}
\begin{small}
 \begin{align}
   \begin{split}
    \mathcal{P} = &p_{\emph{E}_\emph{B}}\!(\sigma_\emph{B}, \theta_\emph{B})
                  p_{\emph{E}_\emph{U}}\!(\sigma_\emph{U}, \theta_\emph{U})
                  \!\!\left[\mathbb{P}_\mathcal{D}
                  \!\!\!\int_0^{\infty} 
                  \!\!\!\!\!\exp\!\left(\!\!-\frac{ \tau N_0}{N_\emph{B} N_\emph{U} P\!L_\emph{D}(x)}\!\!\right)
                  \!\mathcal{L}_{I_\emph{D}}^\emph{d}\!\!\left(\!\!-\frac{ \tau}{N_\emph{B} N_\emph{U} P\!L_\emph{D}(x)}\!\!\right)
                  \!\mathcal{L}_{I_\emph{R}}^\emph{d}\!\!\left(\!\!-\frac{ \tau}{N_\emph{B} N_\emph{U} P\!L_\emph{D}(x)}\!\!\right)
                  \!\!f_\emph{D}(x) \emph{d} x \right.\\
                  & \left.+
                  \mathbb{P}_\mathcal{R}
                  \!\!\int_0^{\infty} 
                  \!\!\!\!\exp\!\left(\!\!-\frac{\tau N_0}{N_\emph{B} N_\emph{U} P\!L_\emph{R}(x)}\!\right)
                  \!\mathcal{L}_{I_\emph{D}}^\emph{r}\!\!\left(\!\!-\frac{\tau}{N_\emph{B} N_\emph{U} P\!L_\emph{R}(x)}\!\right)
                  \!\mathcal{L}_{I_\emph{R}}^\emph{r}\!\!\left(\!\!-\frac{\tau}{N_\emph{B} N_\emph{U} P\!L_\emph{R}(x)}\!\right)
                  \!\!f_\emph{R}(x) \emph{d} x
                  \right],
   \end{split}
 \end{align}
\end{small}%
  where the functions $\mathcal{L}_{I_\emph{D}}^\emph{d}(\cdot)$, $\mathcal{L}_{I_\emph{R}}^\emph{d}(\cdot)$, $\mathcal{L}_{I_\emph{D}}^\emph{r}(\cdot)$ and $\mathcal{L}_{I_\emph{R}}^\emph{r}(\cdot)$ are given by
\begin{small} 
 \begin{align}
  \begin{split}
\mathcal{L}_{I_\emph{D}}^\emph{d}\!\!\left(\cdot\right) = 
   \exp \!\left(
   \!\!-\frac{\lambda_\emph{B} \theta_\emph{B}\theta_\emph{U}}{2\pi}
   \!\!\int_{x}^{\infty}
   \!\!\!\frac{\tau  P\!L_\emph{D\!}(u)  p_\emph{L\!}(u) }
   {P\!L_\emph{D\!}(r) + \tau P\!L_\emph{D\!}(u)}
   u \emph{d} u
   \!\!\right),
  \end{split}
 \end{align}
\end{small}%
 \setlength{\abovedisplayskip}{0pt}
\begin{small} 
 \begin{align}
  \begin{split}
    &\mathcal{L}_{I_\emph{R}}^\emph{d}\!\!\left(\cdot\right) = \\
        &\exp \!\!\Bigg(
        \!\!\!-\!2 \pi \lambda_\emph{B}
        \!\!\!\int_0^{\infty}
        \!\!\!\Bigg(\!\!1\!-\! \exp \!\!\Big(
        \!\!\!-\!  \frac{\lambda_\emph{R}\theta_\emph{B}\theta_\emph{U}}{4\pi^2}
        \!\!\!\int_{\mathcal{C}_2}
        \!\!\!\!\frac{\tau \!P\!L_\emph{R\!}(\!\sqrt{\!u^2\!+\!t^2\!-\!2 u t \!\cos \!\psi}\!+\!t)
        p_\emph{L\!}(\!\sqrt{\!u^2\!+\!t^2\!-\!2 u t \!\cos \!\psi}) p_\emph{L\!}(t) 
        p_\emph{F}(u,\!t,\!\psi)}
        {P\!L_\emph{D\!}(x)+\tau P\!L_\emph{R\!}(\sqrt{\!u^2+t^2-2 u t \cos \psi}\!+\!t)}
        t \emph{d} t \emph{d} \psi
        \!\Big)\!\!\!\Bigg)
        \!u \emph{d} u
        \!\!\Bigg),                
  \end{split}
 \end{align}
\end{small}%
\setlength{\abovedisplayskip}{0pt}
\begin{small}
 \begin{align}
  \begin{split}
  \mathcal{L}_{I_\emph{D}}^\emph{r}\!\!\left(\cdot\right) =
      \exp \!\left(
      \!\!-\frac{\lambda_\emph{B} \theta_\emph{B}\theta_\emph{U}}{2\pi}
      \!\!\int_{x \gamma^{-\frac{1}{\alpha}}}^{\infty}
      \!\frac{\tau  P\!L_\emph{D\!}(u)  p_\emph{L\!}(u) }
      {P\!L_\emph{R\!}(x) + \tau P\!L_\emph{D}(u)}
      u \emph{d} u
      \!\!\right),                      
  \end{split}
 \end{align}
\end{small}%
\setlength{\abovedisplayskip}{0pt}
\begin{small}
\begin{align}
  \begin{split}
  &\mathcal{L}_{I_\emph{R}}^\emph{r}\!\!\left(\cdot\right) =\\
   &\!\exp \!\!\Bigg(
   \!\!\!-\!2 \pi \lambda_\emph{B}
   \!\!\!\int_0^{\infty}
   \!\!\!\Bigg(\!\!1\!-\! \exp \!\!\Big(
   \!\!\!-\!  \frac{\lambda_\emph{R}\theta_\emph{B}\theta_\emph{U}}{4\pi^2}
   \!\!\!\int_{\mathcal{C}_3} 
   \!\!\!\!\frac{\tau \!P\!L_\emph{R\!}(\!\sqrt{\!u^2\!+\!t^2\!-\!2 u t \!\cos \!\psi}\!+\!t)
   p_\emph{L\!}(\!\sqrt{\!u^2\!+\!t^2\!-\!2 u t \!\cos \!\psi}) 
   p_\emph{L\!}(t) 
   p_\emph{F}(u,\!t,\!\psi)}
   {P\!L_\emph{R\!}(x)+\tau P\!L_\emph{R\!}(\sqrt{\!u^2+t^2-2 u t \cos \psi}\!+\!t)}
   t \emph{d} t \emph{d} \psi
   \Big)\!\!\!
   \Bigg)
   \!u \emph{d} u
   \!\!\Bigg),
  \end{split}
\end{align}
\end{small}%
 where the regions $\mathcal{C}_2$ and $\mathcal{C}_3$ consist of two sub-regions $\mathcal{C}_{21}$ and $\mathcal{C}_{22}$, $\mathcal{C}_{31}$ and $\mathcal{C}_{32}$, respectively.
\begin{small} 
 \begin{align}
  \begin{split}
 \mathcal{C}_{21} 
 \!=\! &\left\{ 
   \!t, \psi \!:  
   t \!\in\!\! 
   \left(\!\!\max \!\left\{\!\frac{u^2 - (r_\text{D}\gamma^{\frac{1}{\alpha}})^2}
   {2(u \cos \psi \!-\! r_\text{D}\gamma^{\frac{1}{\alpha}}\!)}, 0 \!\right\}\!,\infty \!\!\right)\!\!,
   \psi \!\in\!\! 
   \left(\!\!-\pi, -\!\arccos\!\!\left(\!\!\frac{r_\text{D}\gamma^{\frac{1}{\alpha}}}{u} \!\!\right)\!\!\!\right) 
   \!\cup\! 
   \left(\!\!\arccos\!\!\left(\!\!\frac{r_\text{D}\gamma^{\frac{1}{\alpha}}}{u}\!\!\right)\!\!, 
   \pi\!\right)\!
   \!\right\},
  \end{split}
 \end{align}
\end{small}%
\setlength{\abovedisplayskip}{0pt}
\begin{small}
 \begin{align}
  \begin{split}
 \mathcal{C}_{22} 
 \!=\! &\left\{ 
   \!t, \psi \!:  
   t \!\in\!\! 
   \left(0,\frac{u^2 - (r_\text{D}\gamma^{\frac{1}{\alpha}})^2}
   {2(u \cos \psi \!-\! r_\text{D}\gamma^{\frac{1}{\alpha}}\!)}\!\right)\!\!,
   \psi \!\in\!\! 
   \left(\!-\arccos\!\left(\!\frac{r_\text{D}\gamma^{\frac{1}{\alpha}}}{u} \!\right), \arccos\!\left(\!\frac{r_\text{D}\gamma^{\frac{1}{\alpha}}}{u}\!\right)\!\!\right) 
   \!\right\},
  \end{split}
 \end{align}
\end{small}%
\setlength{\abovedisplayskip}{0pt}
\begin{small}
 \begin{align}
  \begin{split}
 \mathcal{C}_{31} 
 \!=\! &\left\{ 
   t, \psi \!:  
   t \!\in\!\! 
   \left(\!\max \!\left\{\!\frac{u^2 - r_\text{R}^2}
   {2(u \cos \psi \!-\! r_\text{R})}, 0 \!\right\}\!,\infty \!\right)\!,
   \psi \in\! 
   \left(\!-\pi, -\!\arccos\!\left(\frac{r_\text{R}}{u}\right)\!\right) 
   \!\cup\! 
   \left(\!\arccos\!\left(\frac{r_\text{R}}{u}\right)\!, 
   \pi\!\right)\!
   \!\right\},
  \end{split}
 \end{align}
\end{small}%
\setlength{\abovedisplayskip}{0pt}
\begin{small}
 \begin{align}
  \begin{split}
 \mathcal{C}_{32} 
 \!=\! &\left\{ 
   t, \psi \!:  
   t \!\in\! 
   \left(0,\frac{u^2 - r_\text{R}^2}
   {2(u \cos \psi \!-\! r_\text{R})}\!\right)\!,
   \psi \!\in\!
   \left(\!-\arccos\!\left(\frac{r_\text{R}}{u}\right), \arccos\!\left(\frac{r_\text{R}}{u}\right)\!\right) 
   \!\right\}.
  \end{split}
 \end{align}
\end{small}%
 The functions $\mathbb{P}_\mathcal{D}$, $\mathbb{P}_\mathcal{R}$, $f_\emph{D}(\cdot)$ and $f_\emph{R}(\cdot)$ are given in Eq. (\ref{XD}), (\ref{XR}), (\ref{PDF_D}) and (\ref{PDF_R}), respectively.
 \end{thm}
 \begin{proof}
 See Appendix A.
 \end{proof}
 \vspace*{-13pt}
 \hrulefill
 \vspace*{-10pt}
 \end{figure*}
The coverage probability of the network consists of three terms
\begin{align}\label{CP_define}
  \mathcal{P} 
  =& \mathbb{P}_\mathcal{D} \cdot \mathbb{P}[\text{SINR}\!>\!\tau|\mathcal{D}] 
    \!+ \!\mathbb{P}_\mathcal{R} \cdot \mathbb{P}[\text{SINR}\!>\!\tau|\mathcal{R}] 
    \!+ \!\mathbb{P}_\mathcal{O}\cdot 0,
\end{align}
 where $\mathbb{P}[\text{SINR}\!>\!\tau|\mathcal{D}]$ and $\mathbb{P}[\text{SINR}\!>\!\tau|\mathcal{R}]$ are the coverage probability conditioned on that the typical user associates to a BS through a direct link and a reflected link, respectively. Specifically, when the typical user has no BS to associate, the conditional coverage probability is zero. The full expression of the coverage probability is displayed in Theorem \ref{thm1} on the next page.

 We observe that the amount of the resources used in channel estimation, i.e., the value of the unit training overhead $\beta$, primarily affects the coverage probability on the $p_{\text{E}_\text{B}}(\sigma_\text{B}, \theta_\text{B}) p_{\text{E}_\text{U}}(\sigma_\text{U}, \theta_\text{U})$ term. In this reason, we can simplify the relationship between the coverage probability $\mathcal{P}$ and the unit training overhead $\beta$, leading to the following Corollary.

\begin{col}\label{col1}
The coverage probability $\mathcal{P}$ in Theorem \ref{thm1} is a monotonically increasing function of the unit training overhead $\beta$.
\end{col}
\begin{proof}
See Appendix B.  
\end{proof}

On the one hand, this result reveals that, to achieve the optimal coverage performance, the resources allocated to channel estimation should be as much as possible. Owing to that the penalty of redundant channel estimation resources only affects the ASE and the EE, the coverage performance merely enjoys the benefits of abundant channel estimation resources. On the other hand, the coverage-optimal resource allocation scheme has nothing to do with network deployment, i.e., $\lambda_\text{B}$ and $\lambda_\text{R}$, which enlightens a separated design of network deployment and resource allocation.

\subsection{Area Spectrum Efficiency}
As defined in Section \ref{performmetric}, the ASE $\mathcal{A}$ is expressed as
\begin{align}\label{ASE}
 \begin{split}
\mathcal{A} 
=& \frac{T_\text{D}}{T} \lambda_\text{B} \mathbb{E}[\text{log}_2 (1+\text{SINR}) \mathbb{I}\{\text{SINR}>\tau\}]\\
=& \frac{T\!-\!\beta(M_\text{B}M_\text{R}M_\text{U}\!+\!M_\text{B}M_\text{U})}{T} \lambda_\text{B}
  \!\!\left(\!\frac{1}{\ln\!2} 
  \!\!\int_{\ln(1\!+\tau)}^{\infty} 
  \!\!\!\!\!\!\mathbb{P}[\text{SINR}\!>\!e^t\!\!-\!1] \text{d} t 
  \!+ \!\text{log}_2(1\!+\!\tau) \mathbb{P}[\text{SINR}\!>\!\tau]\!\!\right).
  \end{split}
\end{align}
The full expression of $\mathcal{A}$ can be obtained by applying the results in Theorem \ref{thm1} to Eq. (\ref{ASE}). We neglect it due to page limitation. Similar to the conclusions on the coverage probability, we propose the following Corollary on the ASE-optimal unit training overhead $\beta^*_\mathcal{A}$.

\begin{col}\label{col2}
The optimal unit training overhead $\beta^*_\mathcal{A}$ that maximizes the ASE $\mathcal{A}$ satisfies the following condition
\begin{align}\label{betaASE}
  \begin{split}
 &\frac{\emph{SNR}}{\sqrt{\pi} f(\beta^*_\mathcal{A})}
 \!\Big(\frac{T}{M_\emph{B}M_\emph{R}M_\emph{U}\!+\!M_\emph{B}M_\emph{U}}\!-\!\beta^*_\mathcal{A}\!\Big)
  \!\!\Bigg(\!\!
 \frac{a e^{-a^2 \!f^2\!(\beta^*_\mathcal{A})}} {\emph{erf}(a f\!(\beta^*_\mathcal{A})\!)} \!+\!
 \frac{b e^{-b^2 \!f^2\!(\beta^*_\mathcal{A})}} {\emph{erf}(b f\!(\beta^*_\mathcal{A})\!)} \!-\!
 \frac{c e^{-c^2 \!f^2\!(\beta^*_\mathcal{A})}} {\emph{erf}(c f\!(\beta^*_\mathcal{A})\!)} \!-\!
 \frac{d e^{-d^2 \!f^2\!(\beta^*_\mathcal{A})}} {\emph{erf}(d f\!(\beta^*_\mathcal{A})\!)}
 \!\!\Bigg) 
 \!\!=\! 1,  
  \end{split}
\end{align}
where $f(\beta^*_\mathcal{A}) = \sqrt{1+\beta^*_\mathcal{A} \emph{SNR}}$, $a= \frac{\theta_\emph{B}}{2\pi \sqrt{2k_\emph{B}}}$, $b= \frac{\theta_\emph{U}}{2\pi \sqrt{2k_\emph{U}}}$, $c= \frac{1}{\sqrt{2k_\emph{B}}}$ and $d= \frac{1}{\sqrt{2k_\emph{U}}}$.
\end{col}
\begin{proof}
See Appendix C. 
\end{proof}

Although we are not able to derive a simpler expression of the ASE-optimal unit training overhead $\beta^*_\mathcal{A}$, the numerical result of $\beta^*_\mathcal{A}$ can be obtained with low complexity when the value of other parameters are given. We will discuss the influence of other parameters on $\beta^*_\mathcal{A}$ in the next section.

From Corollary \ref{col1} and Corollary \ref{col2}, we notice that the optimal unit training overhead to maximize the coverage probability and the ASE are different. As a consequence, how to choose a proper unit training overhead to balance the coverage probability (coverage performance) and the ASE (rate performance) should be considered.

\subsection{Energy Efficiency}
The expression of the EE, expressed as $\mathcal{E} = \mathcal{A}/(\lambda_\text{B} P_\text{B}+ \lambda_\text{R} P_\text{R})$, can be directly obtained from the expression of $\mathcal{A}$. And the EE-optimal unit training overhead is the same as the ASE-optimal unit training overhead, i.e., $\beta^*_\mathcal{E} = \beta^*_\mathcal{A}$. 

\begin{table}[t]
\renewcommand{\arraystretch}{1.3}
\caption{Default Values of Parameters}
\label{Parameters}{}
\centering
\begin{tabular}{c|c|c|c}
\hline\hline
\bfseries Parameter & \bfseries Notation and Value & \bfseries Parameter & \bfseries Notation and Value\\
\hline\hline
Path loss exponent& $\alpha = 4 $ & Reflection power attenuation & $\gamma = 0.85 $\\
\hline
BS density & $\lambda_\text{B} = 10\text{ BSs/km}^2$ & Blockage density & $\lambda_\text{L} = 500\text{ Blockages/km}^2$  \\
\hline
\multirow{3}{*}{Number of antennas or elements} & $ M_\text{B} = 16$ & RIS deployment fraction& $\mu = 0.6$ \\
\cline{2-4}
& $ M_\text{R} = 16$ & Size of blockages & $l = 15\text{ m}$\\
\cline{2-4}
& $ M_\text{U} = 4$ & SINR threshold & $\tau = 3 \text{ dB}$\\
\hline
\multirow{2}{*}{Beam alignment error parameters}
& $k_\text{B} = 0.02$ &Frame length & $T = 4480 \text{ Symbols}$\\
\cline{2-4}
& $k_\text{U} = 0.08$ &Unit training overhead & $\beta = 1 \text{ Symbols/path}$\\
\hline
\multirow{2}{*}{Transmit or reflect power} & $P_\text{B} = 40 \text{ dBm}$ & Normalized noise power & $N_0 = -90\text{ dB}$\\
\cline{2-4}
& $P_\text{R} = 15 \text{ dBm}$ & Average SNR of the channel & $\text{SNR} = 16 \text{ dB}$\\
\hline
\end{tabular}
\end{table}

\section{Numerical Results}\label{secnum}
In this section, we focus on the numerical results of the coverage probability, the ASE and the EE with respect to the unit training overhead $\beta$ and the RIS deployment fraction $\mu$. The properties of the ASE-optimal (also the EE-optimal) unit training overhead $\beta^*_{\mathcal{A}}$ will be discussed under different system settings. The default values of parameters are listed in Table \ref{Parameters}. In particular, the optimal performance metrics are marked with stars in the following figures.

\subsection{Performance w.r.t Resource Allocation}
In this subsection, we study the performance metrics as functions of the unit training overhead $\beta$ under different system settings, i.e., the number of BS antennas $M_\text{B}$, beam alignment error parameters $k_\text{B}$ and $ k_\text{U}$, frame length $T$ and average SNR. In all the cases shown in Fig. \ref{beta_MB}, \ref{beta_kBkU}, \ref{beta_T} and \ref{beta_SNR}, the coverage probability increases monotonically w.r.t. the unit training overhead $\beta$, as we have proved in Corollary \ref{col1}. Since using more resources on channel estimation brings a smaller beamforming error, the association link has a lower probability of beam misalignment and the coverage probability is promoted. As we have mentioned, the feasible region of the unit training overhead is $\beta \in [0,\frac{T}{M_\text{B}M_\text{R}M_\text{U}+M_\text{B}M_\text{U}})$, and as a consequence, the maximum value of $\beta$ varies with the number of BS antennas $M_\text{B}$ and frame length $T$ in Fig. \ref{beta_MB} and \ref{beta_T}, respectively.

\begin{figure}[t]
  \setlength{\abovecaptionskip}{0pt}
  \setlength{\belowcaptionskip}{-20pt}
  \centering
  \vspace{0in}
  \subfigure[Coverage probability]{\label{CP_beta_MB}
    \includegraphics[width=0.48\textwidth]{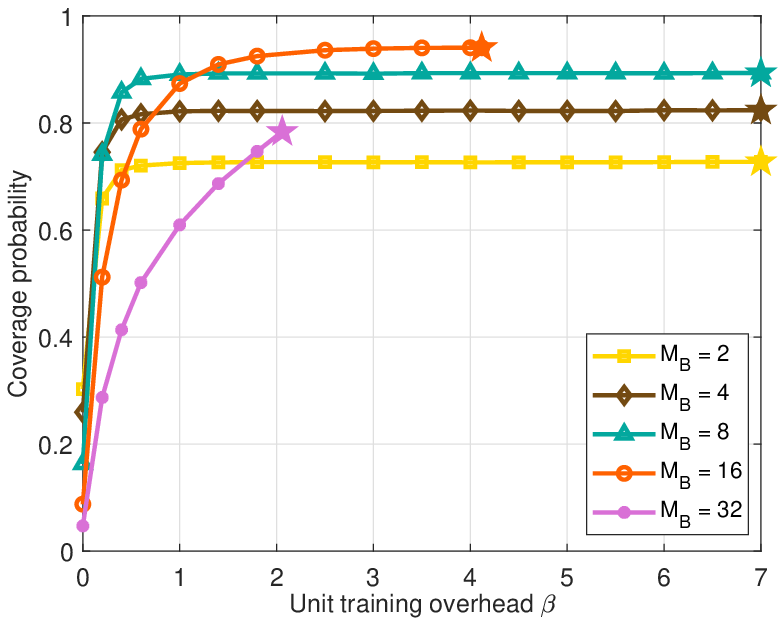}}
  \subfigure[ASE and EE]{\label{ASE_beta_MB} 
    \includegraphics[width=0.485\textwidth]{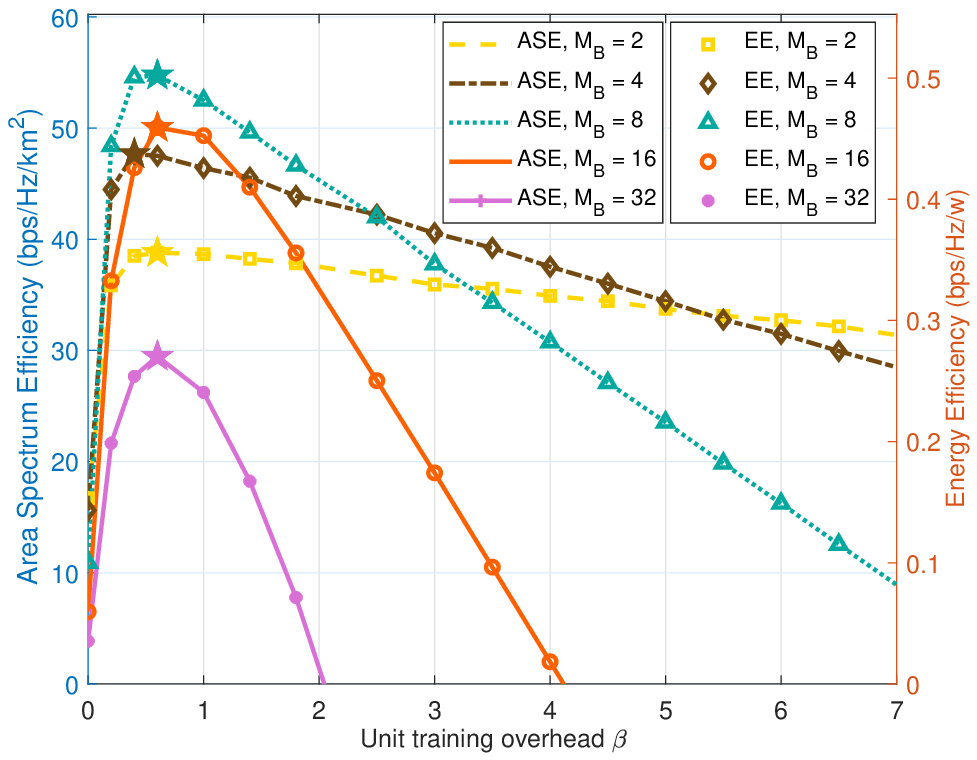}}
  \caption{Coverage probability, ASE and EE scaling with the unit training overhead $\beta$ under different number of BS antennas $M_\text{B}$.}
  \label{beta_MB}
\end{figure}

We first evaluate the performance metrics w.r.t the unit training overhead $\beta$ affected by the number of BS antennas $M_\text{B}$ in Fig. \ref{beta_MB}. Since the parameters of BSs and UEs are symmetric in our model, the influence of the number of UE antennas is analogous. Generally, more antennas result in a narrower beam and a higher probability of beam alignment error, while the beamforming gain is also higher. When the number of BS antennas is relatively small, e.g., $M_\text{B} = 2, 4, 8$, the increasing beamforming gain dominates the performance when $M_\text{B}$ increases, leading to increasing coverage probability in Fig. \ref{CP_beta_MB}. However, when the number of BS antennas is large, e.g., $M_\text{B} = 16, 32$, the severe beam misalignment dominates the performance and increasing $M_\text{B}$ leads to overall decreasing coverage probability in Fig. \ref{CP_beta_MB}. In Fig. \ref{ASE_beta_MB}, the ASE and the EE have the same tendency, since the ratio of ASE to EE, i.e., $\mathcal{A}/\mathcal{E}= \lambda_\text{B} P_\text{B}+ \lambda_\text{R} P_\text{R}$, merely relies on network deployment and energy consumption. Similar to the influence on the coverage probability, increasing the number of BS antennas $M_\text{B}$ first increases and then decreases the ASE and the EE. Moreover, a system equipped with more BS antennas is more sensitive to the change of resource allocation, i.e., the variations of the ASE and the EE are more dramatical w.r.t unit training overhead $\beta$ under large $M_\text{B}$. The optimal unit training overhead $\beta^*_{\mathcal{A}}$ to maximize the ASE and the EE, marked with stars, first decreases, then increases and finally decreases w.r.t the number of BS antennas $M_\text{B}$, resulting from the trade-off between beam alignment error and beamforming gain.

\begin{figure}[t]
  \setlength{\abovecaptionskip}{0pt}
  \setlength{\belowcaptionskip}{-20pt}
  \centering
  \vspace{0in}
  \subfigure[Coverage probability]{\label{CP_beta_kBkU}
    \includegraphics[width=0.48\textwidth]{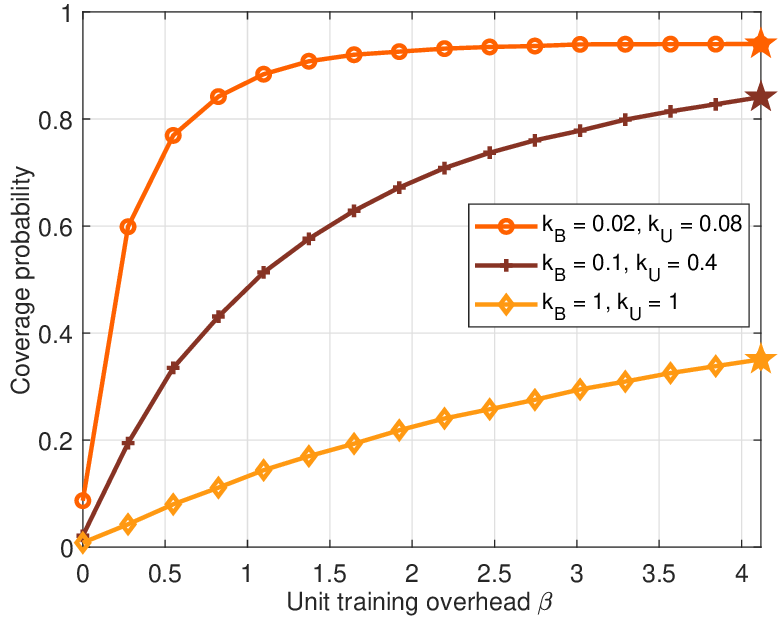}}
  \subfigure[ASE and EE]{\label{ASE_beta_kBkU} 
    \includegraphics[width=0.485\textwidth]{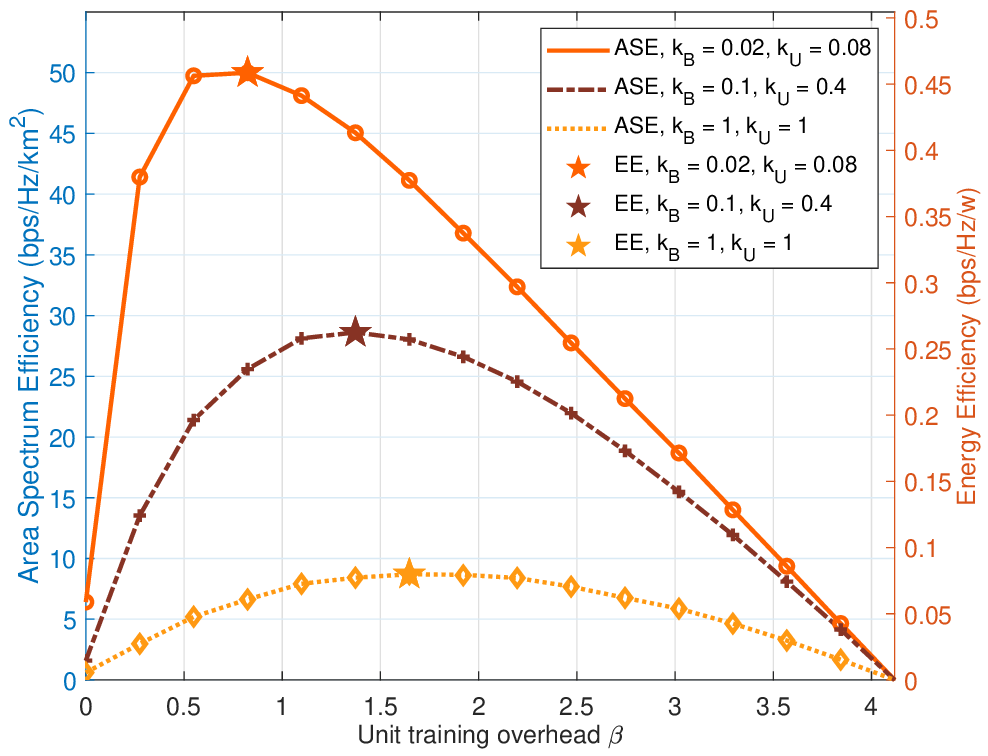}}
  \caption{Coverage probability, ASE and EE scaling with the unit training overhead $\beta$ under different beam alignment error parameters $k_\text{B}$ and $k_\text{U}$.}
  \label{beta_kBkU}
\end{figure}

\begin{figure}[t]
  \setlength{\abovecaptionskip}{0pt}
  \setlength{\belowcaptionskip}{-20pt}
  \centering
  \vspace{0in}
  \subfigure[Coverage probability]{\label{CP_beta_T}
    \includegraphics[width=0.474\textwidth]{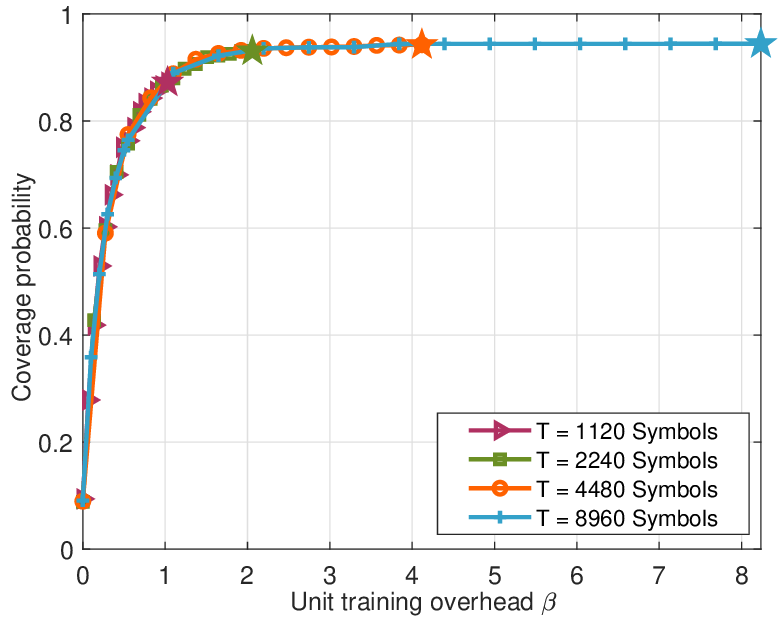}}
  \subfigure[ASE and EE]{\label{ASE_beta_T} 
    \includegraphics[width=0.492\textwidth]{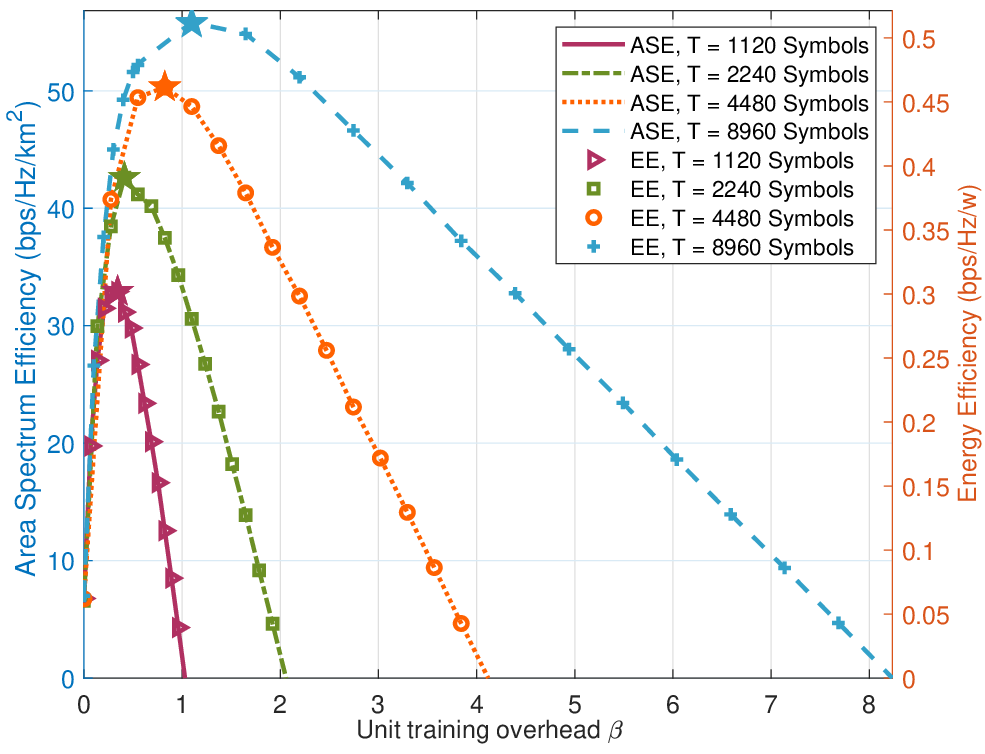}}
  \caption{Coverage probability, ASE and EE scaling with the unit training overhead $\beta$ under different frame length $T$.}
  \label{beta_T}
\end{figure}

The performance under different levels of beam alignment errors of BSs and UEs, i.e., $k_\text{B}$ and $k_\text{U}$, are displayed in Fig. \ref{beta_kBkU}. Recalling that, the terms $k_\text{B} \pi^2$ and $k_\text{U} \pi^2$ are the variances of the beam alignment error without channel estimation for BSs and UEs, respectively. It is shown that the beam alignment error has more impact on system performance than the number of antennas. Larger $k_\text{B}$ and $k_\text{U}$ represent larger variances of the beam alignment error. Consequently, the coverage probability, the ASE and the EE have huge degradation as beam alignment error parameters $k_\text{B}$ and $k_\text{U}$ increase. In addition, larger $k_\text{B}$ and $k_\text{U}$ lead to larger $\beta^*_{\mathcal{A}}$, indicating that more resources for channel estimation should be allocated in order to compensate for the loss brought by the beam alignment error.

Fig. \ref{beta_T} illustrates the performance under different frame length $T$. The results reveal that increasing frame length does not influence the coverage performance in Fig. \ref{CP_beta_T}, except for the region of $\beta$. On the contrary, a longer frame length brings better ASE and EE performance in Fig. \ref{ASE_beta_T}. And the optimal unit training overhead $\beta^*_{\mathcal{A}}$ that maximizes the ASE and the EE also increases monotonically with frame length $T$. The reason is that a longer frame length $T$ leads to wider space for data transmission, given the same channel estimation overhead. Moreover, the ASE and the EE have the same tendency at the same $T$.

\begin{figure}[t]
  \setlength{\abovecaptionskip}{0pt}
  \setlength{\belowcaptionskip}{-20pt}
  \centering
  \vspace{0in}
  \subfigure[Coverage probability]{\label{CP_beta_SNR}
    \includegraphics[width=0.48\textwidth]{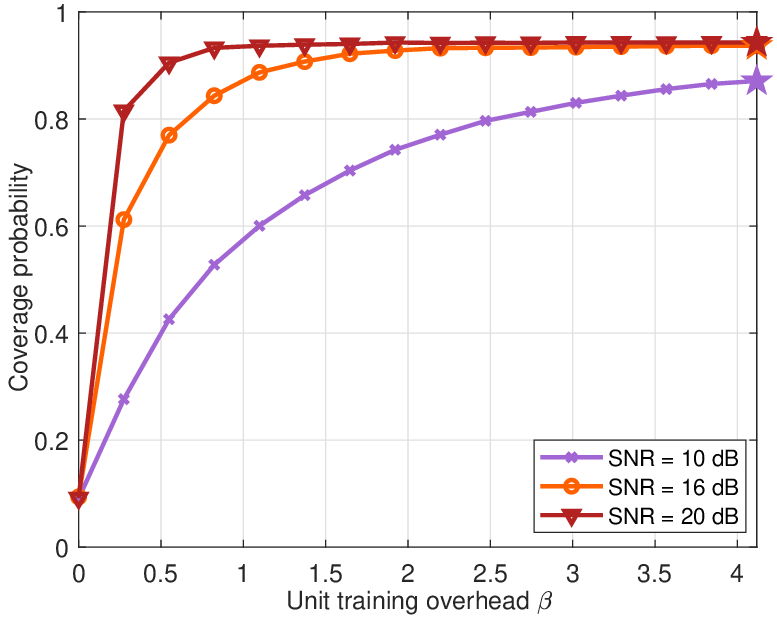}}
  \subfigure[ASE and EE]{\label{ASE_beta_SNR} 
    \includegraphics[width=0.485\textwidth]{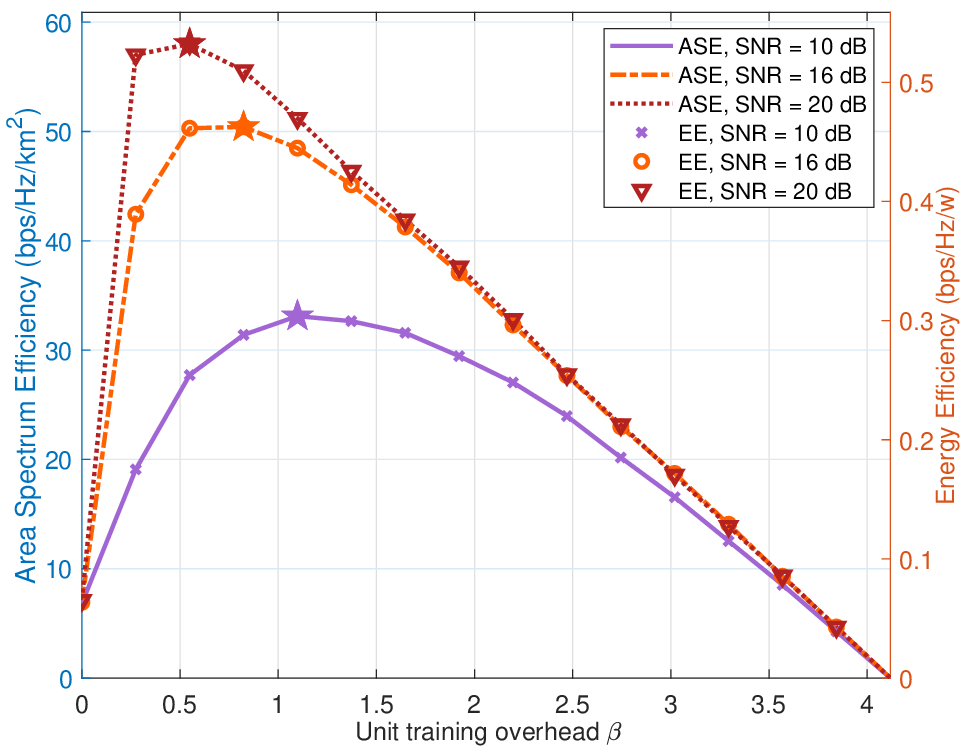}}
  \caption{Coverage probability, ASE and EE scaling with the unit training overhead $\beta$ under different average SNR.}
  \label{beta_SNR}
\end{figure}

The impact of average SNR is displayed in Fig. \ref{beta_SNR}. Naturally, a higher average SNR results in a smaller channel estimation error and thereby better performance in all the cases. However, further increasing average SNR provides marginal gain and the performance is approaching its limit, especially for the coverage probability shown in Fig. \ref{CP_beta_SNR}. Besides, in Fig. \ref{ASE_beta_SNR}, we find that the ASE-optimal unit training overhead $\beta_\mathcal{A}^*$ decreases as average SNR increases, because better channel environment consumes less channel estimation resources to reach its maximum.

\subsection{The Properties of the Optimal Unit Training Overhead $\beta_\mathcal{A}^*$}
Then we focus on the ASE-optimal (as well as the EE-optimal) unit training overhead $\beta_\mathcal{A}^*$ derived in Corollary \ref{col2}. The influence of the number of BS antennas $M_\text{B}$, beam alignment error parameter $k_\text{B}$, frame length $T$ and average SNR are discussed.

\begin{figure}[t]
  \setlength{\abovecaptionskip}{0pt}
  \setlength{\belowcaptionskip}{-20pt}
  \centering
  \vspace{0in}
  \subfigure[The impact of the number of BS antennas $M_\text{B}$ and BS beam alignment error parameter $k_\text{B}$]{\label{Optbeta_kBMB}
    \includegraphics[width=0.47\textwidth]{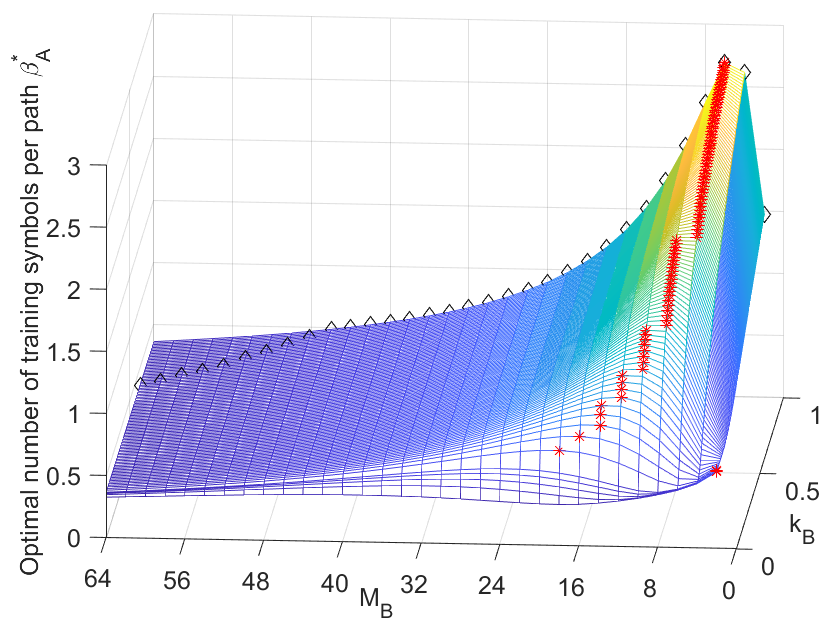}}
  \subfigure[The impact of average SNR and frame length $T$]{\label{Optbeta_SNRandT} 
    \includegraphics[width=0.45\textwidth]{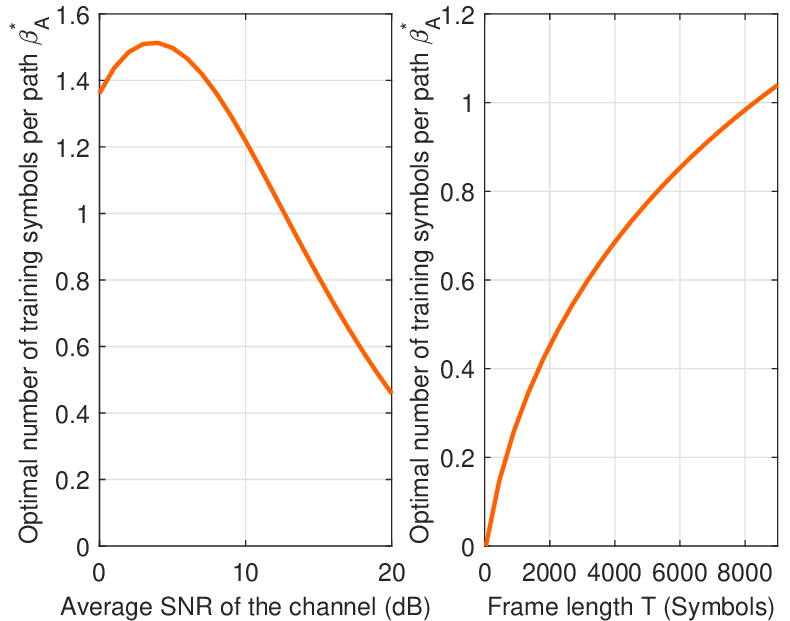}}
  \caption{The impact of some system parameters, i.e., the number of BS antennas $M_\text{B}$, BS beam alignment error parameter $k_\text{B}$, frame length $T$ and average SNR, to the optimal unit training overhead $\beta_\mathcal{A}^*$.}
  \label{Optbeta}
\end{figure}

Fig. \ref{Optbeta_kBMB} shows that two parameters of BS beamforming, i.e., $M_\text{B}$ and $k_\text{B}$, as influencing factors of $\beta_\mathcal{A}^*$. When $k_\text{B}$ is quite small, e.g., $k_\text{B} \leq 0.08$, the optimal unit training overhead $\beta_\mathcal{A}^*$ first decreases, then increases and finally decreases as the number of BS antennas $M_\text{B}$ grows (verified in Fig. \ref{beta_MB}). In contrast, the optimal unit training overhead $\beta_\mathcal{A}^*$  first increases and then decreases when the number of BS antennas $M_\text{B}$ increases under a relatively large $k_\text{B}$, e.g., $0.08< k_\text{B} \leq 1$ in Fig. \ref{Optbeta_kBMB}. In summary, the optimal unit training overhead $\beta_\mathcal{A}^*$ as a function of the number of BS antennas $M_\text{B}$ shows different tendencies in the regions of high beam alignment error (approximately $0.08< k_\text{B} \leq 1$) and low beam alignment error (approximately $k_\text{B} \leq 0.08$). However, in the region of low beam alignment error, the variation of $\beta_\mathcal{A}^*$ between different $M_\text{B}$ is very small, i.e., $\beta_\mathcal{A}^*$ roughly falls in the rang of $(0.2, 1.2)$. Conversely, as $k_\text{B}$ increases, the probability of beam misalignment increases, and the variation of $\beta_\mathcal{A}^*$ between different $M_\text{B}$ becomes more drastic. 

In Fig. \ref{Optbeta_kBMB}, the red stars mark the maximum $\beta_\mathcal{A}^*$ and its corresponding $M_\text{B}$ for each $k_\text{B}$. Firstly, it is suggested that BSs with medium scale of antennas, e.g., $4\leq M_\text{B} \leq 16$, should be allocated more resources on channel estimation, no matter the degree of the beam alignment error. Since in this situation, beam misalignment dominates the performance over beamforming gain. Secondly, when the number of BS antennas is relatively large, e.g., $M_\text{B} > 16$, beamwidth is quite narrow and beam misalignment is always critical, which makes resource allocation less useful, resulting in relatively few channel estimation resources allocation. Finally, when the number of BS antennas is quite small, e.g., $M_\text{B} = 2$, channel estimation resources can be shrunk, because the particularly wide beam is insensitive to beam misalignment. From another point of view, the black diamonds mark the maximum $\beta_\mathcal{A}^*$ and the corresponding $k_\text{B}$ for each $M_\text{B}$. For most cases, e.g., $M_\text{B} \leq 48$, a larger $k_\text{B}$ corresponds to allocating more resources on channel estimation, in order to get a more accurate CSI (verified in Fig. \ref{beta_kBkU}). While when the number of BS antennas is very large, e.g., $M_\text{B} \geq 48$, a larger $k_\text{B}$ may not lead to a larger $\beta_\mathcal{A}^*$ due to the resource saving concern. Yet $\beta_\mathcal{A}^*$ for $M_\text{B} \geq 48$ is insensitive to both $M_\text{B}$ and $k_\text{B}$.

Fig. \ref{Optbeta_SNRandT} illustrates the impact of frame length $T$ and average SNR on the optimal unit training overhead $\beta_\mathcal{A}^*$. In the left figure of Fig. \ref{Optbeta_SNRandT}, the optimal unit training overhead $\beta_\mathcal{A}^*$ first increases and then decreases w.r.t average SNR. In the low SNR region, e.g., $\text{SNR}<4 \text{dB}$, the unit training overhead $\beta$ dominates the variance of beam alignment error ($\sigma_\text{B}^2 = k_\text{B} \pi^2/\sqrt{1+\beta \text{SNR}}$). Therefore, higher SNR requires a larger $\beta_\mathcal{A}^*$ in the low SNR region. On the contrary, in the high SNR region, beam alignment error is low enough and a smaller $\beta_\mathcal{A}^*$ is preferred to reserve more resources for data transmission (verified in Fig. \ref{beta_SNR}). In the right figure of Fig. \ref{Optbeta_SNRandT}, it is shown that longer frame length $T$ corresponds to more resources to be allocated and a larger $\beta_\mathcal{A}^*$ should be chosen (verified in Fig. \ref{beta_T}).

\subsection{Performance w.r.t Network Deployment}
\begin{figure}[t]
  \setlength{\abovecaptionskip}{0pt}
  \setlength{\belowcaptionskip}{-20pt}
  \centering
  \vspace{0in}
  \subfigure[Coverage probability]{\label{CP_miu_lambdaL}
    \includegraphics[width=0.475\textwidth]{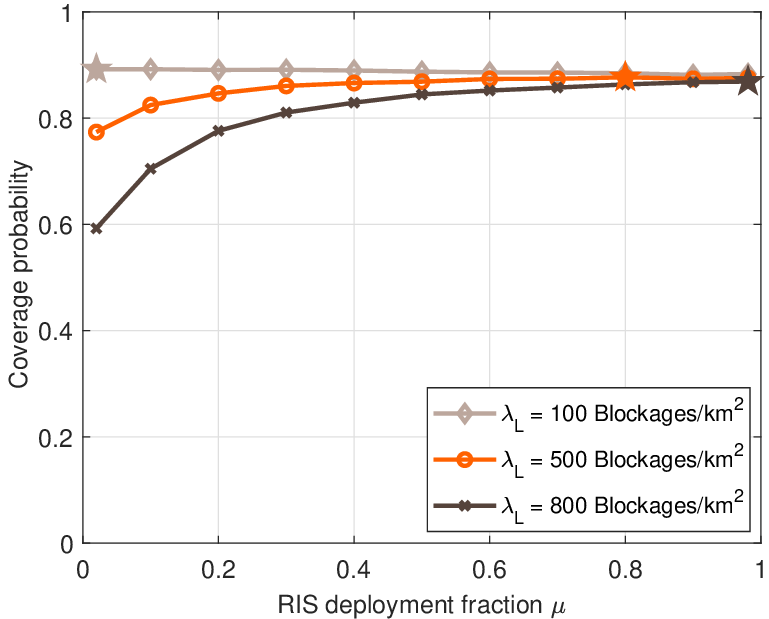}}
  \subfigure[ASE and EE]{\label{ASE_miu_lambdaL} 
    \includegraphics[width=0.485\textwidth]{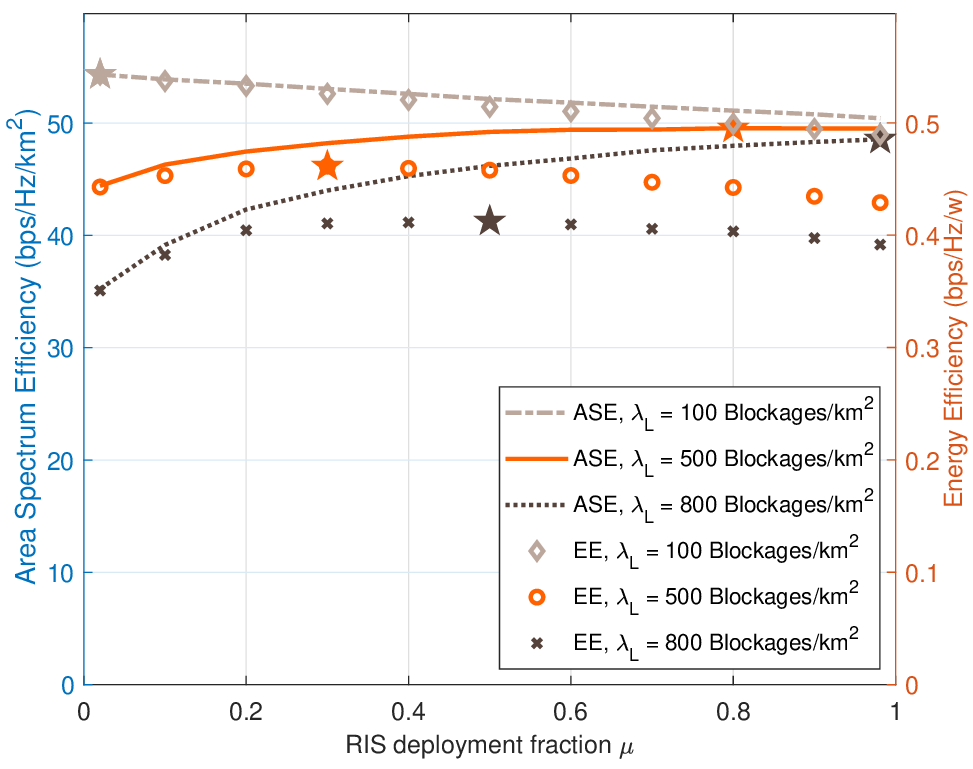}}
  \caption{Coverage probability, ASE and EE scaling with the RIS deployment fraction $\mu$ under different blockage density $\lambda_\text{L}$.}
  \label{miu_lambdaL}
\end{figure}
In this subsection, the performance metrics w.r.t the RIS deployment fraction $\mu$ are evaluated under different blockage densities $\lambda_\text{L}$ and BS densities $\lambda_\text{B}$ to study the optimal network deployment scheme. 

Fig. \ref{miu_lambdaL} shows the performance under different blockage densities $\lambda_\text{L}$. Overall, high density of blockages harms the coverage probability, the ASE and the EE. And the optimal RIS deployment fraction increases as the blockage density increases. When the blockage density is small, e.g. $\lambda_\text{L} = 100 \text{ Blockages/km}^2$, most of the UEs associate to serving BSs through direct links. Consequently, deploying more RISs provides minor gain on desired signal enhancement, but the total interference is greatly enlarged by reflection. Accordingly, the overall performance even degrades as more RISs are deployed. However, when the network has intensive blockages and most of the BS-UE links are blocked, e.g. $\lambda_\text{L} = 500 \text{ Blockages/km}^2$ and $\lambda_\text{L} = 800 \text{ Blockages/km}^2$, deploying appropriate amount of RISs enhances the performance notably as RISs provide additional communication links to NLoS UEs, and the impact of reflected interference is well compensated. Fig. \ref{ASE_miu_lambdaL} indicates that the ASE-optimal RIS deployment fraction may not equals to the EE-optimal RIS deployment fraction. And this result reveals that trade-offs between the ASE and the EE performance should be considered when designing RIS deployment schemes. In this case, the ASE-optimal RIS deployment fraction significantly exceeds the EE-optimal RIS deployment fraction. Especially, the EE performance highly depends on the ratio $P_\text{B}/P_\text{R}$, and thus we should deploy more RISs when deploying RISs is `energy-economical'.

\begin{figure}[t]
  \setlength{\abovecaptionskip}{0pt}
  \setlength{\belowcaptionskip}{-20pt}
  \centering
  \vspace{0in}
  \subfigure[Coverage probability]{\label{CP_miu_lambdaB}
    \includegraphics[width=0.473\textwidth]{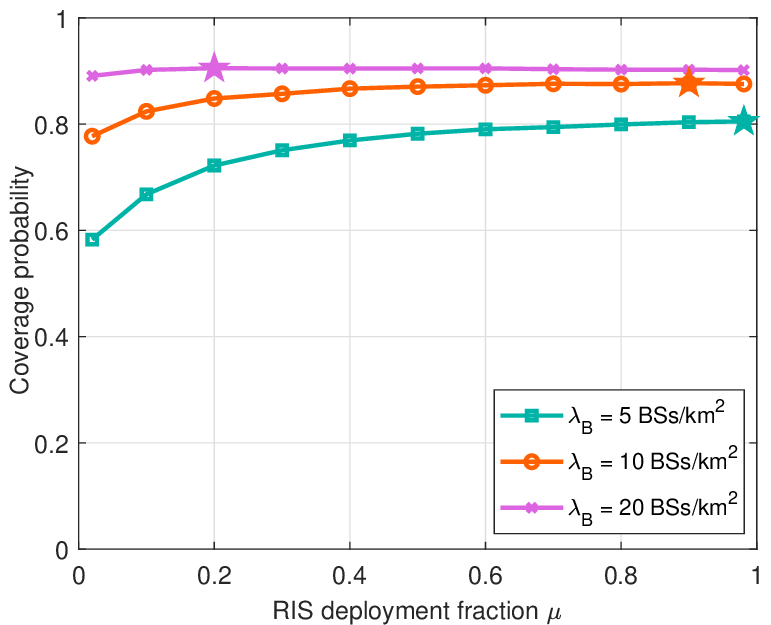}}
  \subfigure[ASE and EE]{\label{ASE_miu_lambdaB} 
    \includegraphics[width=0.485\textwidth]{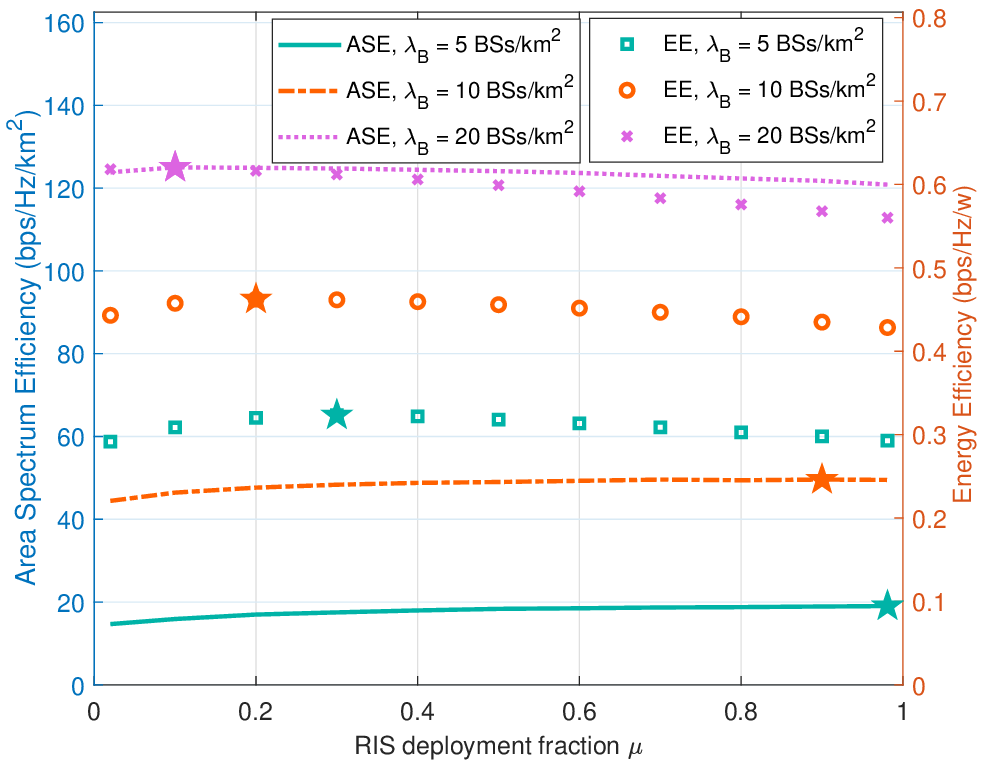}}
  \caption{Coverage probability, ASE and EE scaling with the RIS deployment fraction $\mu$ under different BS density $\lambda_\text{B}$.}
  \label{miu_lambdaB}
\end{figure}

Finally, the performance under varying BS density $\lambda_\text{B}$ is displayed in Fig. \ref{miu_lambdaB}. Deploying denser BSs enhances the coverage probability, the ASE and the EE as the average UE-BS distance is shortened. However, deploying more RISs brings performance degradation under dense BS deployment, especially for the EE in Fig. \ref{ASE_miu_lambdaB}, since BS-UE direct links dominate the association links. The optimal RIS deployment fraction is a monotonically decreasing function of the BS density. This inspires a potential way of replacing part of BSs with RISs for a lower cost, if the performance loss is within tolerance. To summarize Fig.\ref{miu_lambdaL} and \ref{miu_lambdaB}, the greater the ratio of blockage density to BS density, the better the performance of deploying RISs, and the denser RISs should be.

\section{Conclusion}\label{seccon}
This paper has investigated the coverage probability, the ASE and the EE of a downlink RIS-aided multi-cell network with directional transmissions, where the optimal resource allocation and RIS deployment for network performance have been studied. The reflection of interference by RISs, which has not received enough attention, is especially considered. A general model on resource allocation between channel estimation and data transmission is introduced, and the relationship between the beam alignment error and the channel estimation overhead is characterized. Based on these models, the analytical expressions of the coverage probability, the ASE and the EE are derived. Moreover, the monotonicity of the coverage probability w.r.t. the unit training overhead is proved. And the properties of the optimal unit training overhead w.r.t other system parameters, i.e., the number of antennas, beam alignment error parameters, frame length and average SNR, are revealed. Numerical results indicate that the reflection of interference is notable and more RISs are needed only when the ratio of blockage density to BS density is large. However, replacing part of BSs with RISs may result in performance loss, even with lower cost. In the future, we will design more practical interference cancellation and beamforming schemes to alleviate the interference brought by RISs. Methods to reduce the channel estimation overhead in the RIS-aided multi-cell networks with directional transmissions are also worth studying.

\section*{Appendix A: Proof of Theorem 1}
 For the first conditional coverage probability term $\mathbb{P}[\text{SINR}\!>\!\tau|\mathcal{D}]$ in Eq. (\ref{CP_define}), we have the following expressions
\begin{align}\label{SINRD}
  \begin{split}
    \mathbb{P}[\text{SINR}\!>\!\tau|\mathcal{D}]
   =&\mathbb{E}_{r_\text{D}}\Bigg[
   \mathbb{P}\!\left[\frac{G_0 h  P\!L_\text{D}(r_\text{D})}
   {N_0+ I_\text{D}^\text{d}+ I_\text{R}^\text{d}}>\tau
   \Bigg|r_\text{D}\right]\Bigg]\\
   \overset{\text{(a)}}=&
   \int_0^{\infty} \!\mathbb{P}\!\left[\frac{G_0 h  P\!L_\text{D}(r_\text{D})}
   {N_0 + I_\text{D}^\text{d}+ I_\text{R}^\text{d}}>\tau
   \Bigg|r_\text{D}\right] 
   \!f_\text{D}(r_\text{D}) \text{d} r_\text{D}\\
   \overset{\text{(b)}}=&
   \!\!\!\int_0^{\infty} \!\!\!
   p_{\text{E}_\text{B}}\!(\sigma_\text{B}, \theta_\text{B})
   p_{\text{E}_\text{U}}\!(\sigma_\text{U}, \theta_\text{U})
   \mathbb{P}\!\left[\frac{N_\text{B} N_\text{U} h  P\!L_\text{D}(r_\text{D})}
   {N_0 + I_\text{D}^\text{d}+ I_\text{R}^\text{d}}\!\!>\!\tau
   \Bigg|r_\text{D}\!\right] 
   \!\!f_\text{D}(r_\text{D}) \text{d} r_\text{D}\\
   \overset{\text{(c)}}=&    
   \!\!\!\int_0^{\infty}\!\!\!
   p_{\text{E}_\text{B}}\!(\sigma_\text{B}, \theta_\text{B})
   p_{\text{E}_\text{U}}\!(\sigma_\text{U}, \theta_\text{U})
   \exp( - s_1 N_0)
   \mathcal{L}_{I_\text{D}^\text{d}}(s_1)
   \mathcal{L}_{I_\text{R}^\text{d}}(s_1)
   f_\text{D}(r_\text{D}) \text{d} r_\text{D},
  \end{split}
\end{align}
where $s_1 = \frac{\tau}{N_\text{B} N_\text{U} P\!L_\text{D}(r_\text{D})}$ for simplicity, and $r_\text{D}$ is the link distance between the nearest LoS BS and the typical user, while $I_\text{D}^\text{d}$ and $I_\text{R}^\text{d}$ are the normalized direct interference and the normalized reflected interference under a direct association link, respectively. The variables $\mathcal{L}_{I_\text{D}^\text{d}}(s_1)$ and $\mathcal{L}_{I_\text{R}^\text{d}}(s_1)$ indicate the Laplace transforms of $I_\text{D}^\text{d}$ and $I_\text{R}^\text{d}$ evaluated at $s_1$, respectively. Equation (a) follows from the total probability theorem and the PDF of $r_\text{D}$ is given in Eq.(\ref{PDF_D}). Equation (b) follows from the fact that $G_0$ is i.i.d. for each link and the distribution of $G_0$ is shown in Eq.(\ref{G0}). Equation (c) follows from the exponential distribution of $h$.

Then we focus on the Laplace transform term $\mathcal{L}_{I_\text{D}^\text{d}}(s_1)$, which is expressed as
\begin{align}\label{LdD}
  \begin{split}
    \mathcal{L}_{I_\text{D}^\text{d}}(s_1)
   =& \mathbb{E}_{I_\text{D}^{\text{d}}}\!\left[
   \exp\!\left(-s_1 I_\text{D}^{\text{d}}\right)\right]\\
   \overset{\text{(a)}}=& 
   \mathbb{E}_{\phi_\text{B}}\!\!\left[
   \prod_{i:B_i \in \phi_\text{B}\!\backslash\{\!B^*\!\}} 
   \!\!\!\!\!\mathbb{E}_{G_i,h,\mathbb{I}_\text{L}}
   \!\Big[\!\exp\!\left(-s_1 
   \mathbb{I}_\text{L}(r_{B_i})
   P\!L_\text{D}(r_{B_i})
   G_i h \right)\!\Big]\!\right]\\
   \overset{\text{(b)}}=& 
   \mathbb{E}_{\phi_\text{B}}\!\Bigg[
   \!\prod_{i:B_i \in \phi_\text{B}\!\backslash\{\!B^*\!\}}
   \!\!\!\!\!p_\text{L}(r_{B_i})\mathbb{E}_{G_i,h}
   \Big[
   \!\exp(-s_1 
   P\!L_\text{D}(r_{B_i})
   G_i h) 
   \Big] 
   + 1\!-\! p_\text{L}(r_{B_i})\Bigg]\\
   \overset{\text{(c)}}=&
   \mathbb{E}_{\phi_\text{B}}\!\Bigg[
   \!\prod_{i:B_i \in \phi_\text{B}\!\backslash\{\!B^*\!\}}
   \!\!\!\frac{\theta_\text{B}\theta_\text{U}}{4\pi^2}
   p_\text{L}(r_{B_i})
   \mathbb{E}_{h}\Big[
   \!\exp \!\left(\!-s_1
   P\!L_\text{D\!}(r_{B_i})
   N_\text{\!B} N_\text{\!U} h\right)\!\Big]
   + 1- \frac{\theta_\text{B}\theta_\text{U}}{4\pi^2}
   p_\text{L}(r_{B_i}) 
   \!\Bigg]\\
   \overset{\text{(d)}}=&
   \mathbb{E}_{\phi_\text{B}}\!\!\left[
   \!\prod_{i:B_i \in \phi_\text{B}\!\backslash\{\!B^*\!\}}
   \left(1-
   \frac{\theta_\text{B}\theta_\text{U}}{4\pi^2}
   \frac{s_1 p_\text{L\!}(r_{B_i}) P\!L_\text{D\!}(r_{B_i}) N_\text{\!B} N_\text{\!U}}
   {1+s_1 P\!L_\text{D\!}(r_{B_i}) N_\text{\!B} N_\text{\!U}}
   \right)\right]\\
   \overset{\text{(e)}}=&
   \exp \!\left(
   \!-\frac{\lambda_\text{B}\theta_\text{B}\theta_\text{U}}{2\pi} 
   \!\!\int_{r_\text{D}}^{\infty}
   \!\frac{s_1 p_\text{L\!}(u) P\!L_\text{D\!}(u) N_\text{\!B} N_\text{\!U}}
   {1+s_1 P\!L_\text{D\!}(u) N_\text{\!B} N_\text{\!U}}
   u \text{d} u
   \!\right),
  \end{split}
\end{align}
where (a) holds because $G_i,h$ and $\mathbb{I}_\text{L}(r)$ are i.i.d. for each link.  Equations (b) and (c) follow by the distributions of $\mathbb{I}_\text{L}(r)$ and $G_i$ given in Eq. (\ref{IL}) and Eq. (\ref{Gi}). Equation (d) follows from the exponential distribution of $h$ and (e) follows from the PGFL of PPP with variables substitution. 

The Laplace transform term $\mathcal{L}_{I_\text{R}^\text{d}}(s_1)$ can be derived with a similar process and we stretch the major steps
\begin{align}\label{LdR}
  \begin{split}
    &\mathcal{L}_{I_\text{R}^\text{d}}(s_1)\\
   =& \mathbb{E}_{\widehat\phi_\text{B},\widehat\phi_\text{R},G_i,h,\mathbb{I}_\text{L},\mathbb{I}_\text{F}}
   \!\Bigg[
   \exp\!\Bigg(\!\!\!-\!s_1 \!\!\!
   \sum_{i:B_i \in \widehat\phi_\text{B}}
   \sum_{k:R_k \in \widehat\phi_\text{R}}
   \!\!\!\mathbb{I}_\text{L}(\!r_{B_i\!R_k}\!)
   \mathbb{I}_\text{L}(\!r_{R_k}\!)
   \mathbb{I}_\text{F}(\!r_{B_i},r_{R_k},\angle{B_ioR_k}\!)   
   P\!L_\text{R}(\!r_{B_i\!R_k}\!\!+\!r_{R_k}\!)
   G_i h \!\Bigg)\!\Bigg]\\
   =&\mathbb{E}_{\widehat\phi_\text{B},\widehat\phi_\text{R}}\!\Bigg[
   \prod_{i:B_i \in \widehat\phi_\text{B}}
   \prod_{k:R_k \in \widehat\phi_\text{R}}
   \!\!\!\Bigg(\!\!1 \!-\!
   \frac{\theta_\text{B}\theta_\text{U}}{4\pi^2}
   \frac{s_1 
   P\!L_\text{R\!}(\!r_{B_i\!R_k}\!\!+\!r_{R_k}\!) N_\text{\!B} N_\text{\!U}}
   {1\!+\!s_1 \!P\!L_\text{R\!}(\!r_{B_i\!R_k}\!\!+\!r_{R_k}\!) N_\text{\!B} N_\text{\!U}}
   p_\text{L\!}(\!r_{B_i\!R_k}\!) p_\text{L\!}(\!r_{R_k}\!) 
   p_\text{F}(r_{B_i},r_{R_k},\angle{B_ioR_k})\!\Bigg)\!\!\Bigg]\\
   \overset{\text{(a)}}=&
   \exp \!\Bigg(
   \!\!-\!2 \pi \!\lambda_\text{B}
   \!\!\int_0^\infty
   \!\!\Bigg(\!1\!-\! \exp\!\Big(
   \!\!-\!\!  \frac{\lambda_\text{R}\theta_\text{B}\theta_\text{U}}{4\pi^2}
   \!\!\int_{\mathcal{C}_2}
   \frac{s_1 
   P\!L_\text{R\!}\!\left(\!\!\sqrt{\!u^2 \!+\! t^2 \!-\! 2 u t \cos \psi}\!+\!t \!\right) 
   \!N_\text{\!B} N_\text{\!U}}
   {1\!+\!s_1 P\!L_\text{R\!}\!\left(\!\!\sqrt{\!u^2 \!+\! t^2 \!-\! 2 u t \cos \psi}\!+\!t \!\right) 
   \!N_\text{\!B} N_\text{\!U}}\cdot \\
   & \quad \quad p_\text{L\!}\!\left(\!\!\sqrt{\!u^2 \!+\! t^2 \!-\! 2 u t \cos \psi}\right) 
   \!p_\text{L\!}(t) 
   p_\text{F}(u,t,\psi)
   t \text{d} t 
   \text{d} \psi
   \!\Big)\!\!\Bigg)
   u \text{d} u
   \!\!\Bigg),
  \end{split}
\end{align}
where in (a), the BSs and the RISs in $\widehat\phi_\text{B}$ and $\widehat\phi_\text{R}$ are located in the region $\mathcal{C}_2 = \{t, \psi : r_\text{D}^{-\alpha} \geq \gamma (\!\sqrt{\!u^2 \!+\! t^2 \!-\! 2 u t \cos \psi} + t )^{-\alpha}\}$, indicating that interfering links should have weaker average received signal power than the association link. The region $\mathcal{C}_2$ can be rewritten as the union of two sub-regions in Theorem \ref{thm1}. As a consequence, the full expression of $\mathbb{P}[\text{SINR}\!>\!\tau|\mathcal{D}]$ is obtained by combining Eq. (\ref{SINRD}) with Eq. (\ref{PE}), (\ref{PDF_D}), (\ref{LdD}) and (\ref{LdR}).

For the second conditional coverage probability term $\mathbb{P}[\text{SINR}\!>\!\tau|\mathcal{R}]$, we have
\begin{align}\label{SINRR}
  \begin{split}
    \mathbb{P}[\text{SINR}\!>\!\tau|\mathcal{R}] 
   =&\mathbb{E}_{r_\text{R}}\Bigg[
   \mathbb{P}\!\left[\frac{G_0 h  P\!L_\text{R}(r_\text{R})}
   {N_0 + I_\text{D}^\text{r}+ I_\text{R}^\text{r}}>\tau
   \Bigg|r_\text{R}\right]\Bigg]\\
   =&\!\!\!\int_0^{\infty}\!\!\!
   p_{\text{E}_\text{B}}\!(\sigma_\text{B}, \theta_\text{B})
   p_{\text{E}_\text{U}}\!(\sigma_\text{U}, \theta_\text{U})
   \exp( - s_2 N_0)
   \mathcal{L}_{I_\text{D\!}^\text{r}}(s_2)
   \mathcal{L}_{I_\text{R\!}^\text{r}}(s_2)
   f_\text{R}(r_\text{R}) \text{d} r_\text{R},
  \end{split}
\end{align}
where $s_2 = \frac{\tau}{N_\text{B} N_\text{U} P\!L_\text{R}(r_\text{R})}$, while $\mathcal{L}_{I_\text{D}^\text{r}}(s_2)$ and $\mathcal{L}_{I_\text{R}^\text{r}}(s_2)$ indicate the Laplace transforms of $I_\text{D}^\text{r}$ and $I_\text{R}^\text{r}$ evaluated at $s_2$, respectively. For the limited pages, we neglect the derivation of $\mathcal{L}_{I_\text{D}^\text{r}}(s_2)$ and $\mathcal{L}_{I_\text{R}^\text{r}}(s_2)$. The results are
\begin{align}\label{LrD}
  \begin{split}
    \mathcal{L}_{I_\text{D}^\text{r}}(s_2) 
   =&\exp \!\left(
   \!\!-\frac{\lambda_\text{B} \theta_\text{B}\theta_\text{U}}{2\pi}
   \!\!\int_{r_\text{R} \gamma^{-\frac{1}{\alpha}}}^{\infty}
   \!\frac{s_2  p_\text{L\!}(u)  P\!L_\text{D\!}(u) N_\text{B} N_\text{U}}
   {1 + s_2  P\!L_\text{D\!}(u)  N_\text{B} N_\text{U}}
   u \text{d} u
   \!\!\right),
  \end{split}
\end{align}
\begin{align}\label{LrR}
  \begin{split}
    \mathcal{L}_{I_\text{R}^\text{r}}(s_2)
   =&\exp \!\Bigg(
   \!\!-\!2 \pi \!\lambda_\text{B}
   \!\!\int_0^\infty
   \!\!\Bigg(\!1\!-\! \exp\!\Big(
   \!\!-\!\!  \frac{\lambda_\text{R}\theta_\text{B}\theta_\text{U}}{4\pi^2}
   \!\!\int_{\mathcal{C}_3}
   \frac{s_2 
   P\!L_\text{R\!}\!\left(\!\!\sqrt{\!u^2 \!+\! t^2 \!-\! 2 u t \cos \psi}\!+\!t \!\right) 
   \!N_\text{\!B} N_\text{\!U}}
   {1\!+\!s_2 P\!L_\text{R\!}\!\left(\!\!\sqrt{\!u^2 \!+\! t^2 \!-\! 2 u t \cos \psi}\!+\!t \!\right) 
   \!N_\text{\!B} N_\text{\!U}}\cdot \\
   & \quad \quad p_\text{L\!}\!\left(\!\!\sqrt{\!u^2 \!+\! t^2 \!-\! 2 u t \cos \psi}\right) 
   \!p_\text{L\!}(t) 
   p_\text{F}(u,t,\psi)
   t \text{d} t 
   \text{d} \psi
   \!\Big)\!\!\Bigg)
   u \text{d} u
   \!\!\Bigg),
  \end{split}
\end{align}
where $\mathcal{C}_3 = \{t,\psi: \sqrt{\!u^2 \!+\! t^2 \!-\! 2 u t \cos \psi} + t > r_\text{R}\}$ based on the same thought of $\mathcal{C}_2$. And $\mathcal{C}_3$ can also be rewritten as the union of two sub-regions. The full expression of $\mathbb{P}[\text{SINR}\!>\!\tau|\mathcal{R}]$ is obtained by combining Eq. (\ref{SINRR}) with (\ref{PE}), (\ref{PDF_R}), (\ref{LrD}) and (\ref{LrR}). Finally, the expression of $\mathcal{P}$ in Theorem \ref{thm1} is obtained with variables substitution.

\section*{Appendix B: Proof of Corollary 1}
Since the unit training overhead $\beta$ affects the coverage probability in $p_{\text{E}_\text{B}}\!(\sigma_\text{B}, \!\theta_\text{B}\!)p_{\text{E}_\text{U}}\!(\sigma_\text{U}, \!\theta_\text{U}\!)$, which is a multiplier of $\mathcal{P}$ in Theorem \ref{thm1}. We have following equations
\begin{align}
 \begin{split}
    \frac{\text{d}\mathcal{P}}{\text{d} \beta}
  =& K \frac{\text{d}(p_{\text{E}_\text{B}}(\sigma_\text{B}, \theta_\text{B}) p_{\text{E}_\text{U}}(\sigma_\text{U}, \theta_\text{U}))}{\text{d} \beta}\\
  \overset{(\text{a})}=& K \frac{\text{d}
     \left(\frac
     {\text{erf}\left(
     \frac{\theta_\text{B}}{2\pi} \frac{\sqrt{1+\beta \text{SNR}}}{\sqrt{2k_\text{B}}}\right)
     \text{erf}\left(
     \frac{\theta_\text{U}}{2\pi} \frac{\sqrt{1+\beta \text{SNR}}}{\sqrt{2k_\text{U}}}\right)}
     {\text{erf}\left(\frac{\sqrt{1+\beta \text{SNR}}}{\sqrt{2k_\text{B}}}\right)
     \text{erf}\left(\frac{\sqrt{1+\beta \text{SNR}}}{\sqrt{2k_\text{U}}}\right)
     }\right)}
     {\text{d} \beta}\\
  \overset{(\text{b})}=& K \frac{\text{d}
     \left(\frac
     {\text{erf}\left(a f(\beta)\right)
     \text{erf}\left(b f(\beta)\right)}
     {\text{erf}\left(c f(\beta)\right)
     \text{erf}\left(d f(\beta)\right)
     }\right)}
     {\text{d} f(\beta)}
     \frac{\text{d} f(\beta)}{\text{d} \beta}\\    
  \overset{(\text{c})}=& \frac{K \text{SNR}}{2 \sqrt{1+\beta\text{SNR}}} 
     \frac{\text{d}
     \left(\frac
     {\text{erf}\left(a f(\beta)\right)
     \text{erf}\left(b f(\beta)\right)}
     {\text{erf}\left(c f(\beta)\right)
     \text{erf}\left(d f(\beta)\right)
     }\right)}
     {\text{d} f(\beta)},
 \end{split}
\end{align}
where $K$ is a positive constant that represents the multiplier in the coverage probability $\mathcal{P}$ besides the term $p_{\text{E}_\text{B}}(\sigma_\text{B}, \theta_\text{B}) p_{\text{E}_\text{U}}(\sigma_\text{U}, \theta_\text{U})$. Equation (a) follows from the definitions of $\sigma_B^2$ and $\sigma_U^2$ in Section \ref{framestruc}. In (b), we use $f(\beta) = \sqrt{1+\beta \text{SNR}}$, $a= \frac{\theta_\text{B}}{2\pi \sqrt{2k_\text{B}}}$, $b= \frac{\theta_\text{U}}{2\pi \sqrt{2k_\text{U}}}$, $c= \frac{1}{\sqrt{2k_\text{B}}}$ and $d= \frac{1}{\sqrt{2k_\text{U}}}$ for simplicity. Note that, the scope of parameters and variables are: $a\in(0,c]$, $b\in(0,d]$, $c \in (0, +\infty)$, $d \in (0, +\infty)$, $\text{SNR} \in (0, +\infty)$, $\beta \in [0,\frac{T}{M_\text{B}M_\text{R}M_\text{U}+M_\text{B}M_\text{U}})$ and $f(\beta) \in [1,+\infty)$. 

Since the first term in (c) is positive, i.e., $\frac{K \text{SNR}}{2 \sqrt{1+\beta \text{SNR}}} >0$, we focus on the second term
\begin{align}\label{betaCP}
 \begin{split}
     &\frac{\text{d}
     \left(\frac
     {\text{erf}\left(a f(\beta)\right)
     \text{erf}\left(b f(\beta)\right)}
     {\text{erf}\left(c f(\beta)\right)
     \text{erf}\left(d f(\beta)\right)
     }\right)}
     {\text{d} f(\beta)}\\
     =& \frac{2 \text{erf}(a f(\beta)\!) \text{erf}(b f(\beta)\!)}
        {\sqrt{\pi}  f(\beta) \text{erf}(c f(\beta)\!) \text{erf}(d f(\beta)\!)}
        \!\!\left(\!
       \frac{a f\!(\beta) e^{-a^2 \!f^2\!(\beta)}} {\text{erf}(a f\!(\beta)\!)} 
       \!-\!
       \frac{c f\!(\beta) e^{-c^2 \!f^2\!(\beta)}} {\text{erf}(c f\!(\beta)\!)} 
       \!+\!
       \frac{b f\!(\beta) e^{-b^2 \!f^2\!(\beta)}} {\text{erf}(b f\!(\beta)\!)}
       \!-\!
       \frac{d f\!(\beta) e^{-d^2 \!f^2\!(\beta)}} {\text{erf}(d f\!(\beta)\!)}
       \!\right).
  \end{split}
\end{align}
Obviously, the first term outside the brackets is positive. And each term in the formula inside the brackets has the form of $g(x) = x e^{-x^2}/\text{erf}(x)$ with $x>0$. Because $g(x)$ is a monotonically decreasing function of $x$, then $g(a f(\beta))- g(c f(\beta)) \geq 0$ and $g(b f(\beta))- g(d f(\beta)) \geq 0$ hold. Therefore, we have proved that $\frac{\text{d}\mathcal{P}}{\text{d} \beta} >0$ and $\mathcal{P}$ is a monotonically increasing function of $\beta$.

\section*{Appendix C: Proof of Corollary 2}
Since the unit training overhead $\beta$ only influences the ASE on its multiplier term
\begin{align}
 \frac{T\!-\!\beta(M_\text{B}M_\text{R}M_\text{U}\!+\!M_\text{B}M_\text{U})}{T} p_{\text{E}_\text{B}}(\sigma_\text{B}, \theta_\text{B}) p_{\text{E}_\text{U}}(\sigma_\text{U}, \theta_\text{U}). 
\end{align}
 Based on similar thoughts of the proof for Corollary 1, letting $\frac{\text{d}\mathcal{A}}{\text{d} \beta} = 0$ leads to
\begin{align}
 \begin{split}
 &\frac
 {\text{d}\left(
 \left(\!\frac{T}{M_\text{B}M_\text{R}M_\text{U}\!+\!M_\text{B}M_\text{U}}\!-\!\beta\!\right)\!
 p_{\text{E}_\text{B}}(\sigma_\text{B}, \theta_\text{B}) 
 p_{\text{E}_\text{U}}(\sigma_\text{U}, \theta_\text{U})\right)}
 {\text{d} \beta} = 0,
  \end{split}
\end{align}
and thus the final expression in Corollary \ref{col2} holds. 

\bibliographystyle{IEEEtran}
\bibliography{bibfile}
\end{document}